\documentclass[11pt,letterpaper]{article}
\usepackage{jheppub}
\usepackage{comment,xcolor}
\usepackage{amsmath}
\usepackage{mathtools}
\usepackage{amssymb}
\usepackage{amsfonts}
\usepackage{amsthm}
\usepackage{mathrsfs}
\usepackage{footnote}
\usepackage{graphicx} 
\usepackage{subcaption}
\usepackage{stmaryrd}
\usepackage[colorlinks=true,linkcolor=blue,citecolor=blue]{hyperref}

\newtheorem{innercustomgeneric}{\customgenericname}
\providecommand{\customgenericname}{}
\newcommand{\newcustomtheorem}[2]{%
  \newenvironment{#1}[1]
  {%
   \renewcommand\customgenericname{#2}%
   \renewcommand\theinnercustomgeneric{##1}%
   \innercustomgeneric
  }
  {\endinnercustomgeneric}
}

\newcustomtheorem{customthm}{Theorem}

\newtheorem{thm}{Theorem}
\newtheorem{lem}{Lemma}

\newtheorem{cl}{Claim}
\newtheorem{qn}{Question}
\newtheorem{pr}{Proposition}

\newcommand{\dd}{\mathrm{d}}

\newcommand{\ket}[1]{| #1 \rangle}
\newcommand{\bra}[1]{\langle #1 |}
\newcommand{\braket}[2]{{\langle #1 \mid #2 \rangle}}
\newcommand{\ketbra}[2]{{|#1\rangle\hspace{-0.28em}\langle #2|}}

\def\eps{\varepsilon}
\def\O{\mathcal{O}}
\def\h{\mathcal{H}}

\newcommand*{\tr}{\mathrm{Tr}}

\title{The black hole information puzzle and the quantum de Finetti theorem}

\author{Renato Renner and}
\author{Jinzhao Wang }
\affiliation{\small \it Institute for Theoretical Physics, ETH 8093 Z\"urich, Switzerland}
\emailAdd{renner@ethz.ch}
\emailAdd{jinzwang@ethz.ch}

\abstract{
The black hole information puzzle arises from a discrepancy between conclusions drawn from general relativity and quantum theory about the nature of the radiation emitted by a black hole.  According to Hawking's original argument, the radiation is thermal and its entropy thus increases monotonically as the black hole evaporates. Conversely, due to the reversibility of time evolution according to quantum theory, the radiation entropy should start to decrease after a certain time, as predicted by the Page curve. This decrease has been confirmed by new calculations based on the replica trick, which also exhibit its geometrical origin: spacetime wormholes that form between the replicas.  Here we analyse the discrepancy between these and Hawking's original conclusions from a  quantum information theory viewpoint, using in particular the quantum de Finetti theorem. The theorem implies the existence of extra information, $W$, which is neither part of the black hole nor the radiation, but plays the role of a reference. The entropy obtained via the replica trick can then be identified to be the entropy $S(R|W)$ of the radiation conditioned on the reference~$W$, whereas Hawking's original result corresponds to the non-conditional entropy $S(R)$. The entropy $S(R|W)$, which mathematically is an ensemble average, gains an operational meaning in an experiment with $N$ independently prepared black holes: For large $N$, it equals the normalised entropy of their joint radiation, $S(R_1 \cdots R_N)/N$. The discrepancy between this entropy and $S(R)$ implies that the black holes are correlated. The replica wormholes may thus be interpreted as the geometrical representation of this correlation. Our results also suggest a many-black-hole extension of the widely used random unitary model, which we support with non-trivial checks.}

\begin{document}
\maketitle

\section{Introduction} \label{sec:intro}

Black holes are an ideal (theoretical) testbed for exploring the interplay between gravity and quantum physics. Indeed, the phenomenon of Hawking radiation~\cite{hawking1975particle,hawking1976breakdown}, the ``thermal'' radiation emitted by them, can only be explained by combining elements from general relativity and from quantum field theory. While we are still far from a complete understanding of how these two theories fit together to a theory of quantum gravity, which appears to be necessary for a detailed description of black holes, it turned out that a more abstract quantum information-theoretic perspective can yield valuable insights into the nature of Hawking radiation (e.g., \cite{hayden2007black}). In this work we take such an information-theoretic viewpoint. 

The von Neumann entropy, $S(R) = -\tr(\rho_R \log \rho_R)$, of the state $\rho_R$ of the radiation field $R$ emitted by a black hole is a good measure for its ``thermality''. A large value of $S(R)$ indicates that the individual radiation quanta are only little correlated. According to Hawking's calculations, $R$ is indeed thermal in this sense, i.e., the radiation quanta are independent of each other.\footnote{However, the spectrum of the radiation, as measured by an asymptotic observer, deviates  from the Planck form of blackbody radiation by greybody factors~\cite{page1976particle, page2013time}.} This means that $S(R)$ grows with the number of radiation quanta in $R$ and reaches its maximum value when the black hole is evaporated completely. 

But this behavior of $S(R)$ is in conflict with what is sometimes called the \emph{central dogma} of black hole physics~\cite{almheiri2020entropy}. It asserts that a black hole, when regarded from the outside, is a finite-dimensional quantum system and must thus obey the laws of quantum theory, whose equations of motion are fundamentally reversible.  Consequently, if a black hole is formed by collapsing matter that is initially in a pure state, the state of the total system that includes the radiation $R$ must remain pure. While the entropy $S(R)$ may first increase during the radiation process, it must ultimately decrease again and reach zero as the black hole has disappeared completely, corresponding to a final radiation state that is pure. If one additionally assumes that the time evolution is \emph{typical}, in the sense that it is described by a unitary chosen at random according to the Haar measure, one finds what is known today as the \emph{Page curve}\footnote{We plot the Page curve with time $t$ measured in terms of the number of emitted Hawking quanta. A plot with respect to asymptotic time would slightly deviate from this picture, especially at late times, and the Page time, $t_{\mathrm{Page}}$, would lie shortly after the half-life of the black hole~\cite{page2013time}.} and plotted in Fig.~\ref{fig:page1}~\cite{page1993average,page1993information}.

\begin{figure}
    \centering
    \includegraphics[width=0.7\linewidth]{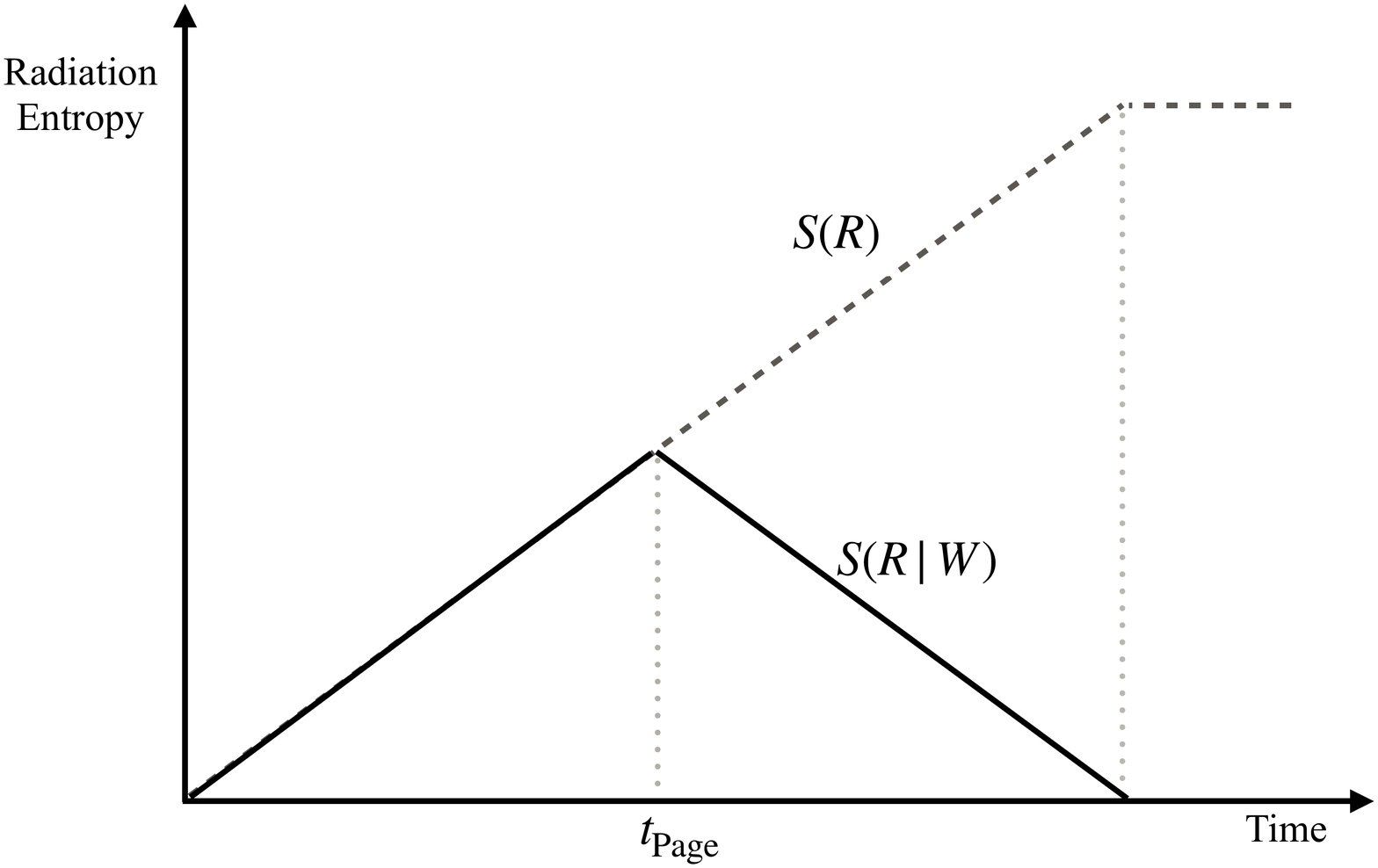}
    \caption{{\bf Black hole information puzzle.} The plot shows the entropy of the  radiation field $R$ (vertical axis) emitted up to a given time (horizontal axis) as predicted by Hawking (dashed line) and Page (solid line). The origin corresponds to the time when the black hole has formed and starts radiating. In a first phase the entropy grows almost linearly in time. According to Hawking's calculations, the growth continues until the black hole is evaporated completely. Conversely, Page concluded that there must be a turning point at a time $t_\mathrm{Page}$, when half of the black hole is evaporated, after which the entropy decreases again until it reaches zero.  The recent calculations based on gravitational path integrals reproduce the Page curve. In this work we argue that the quantum de Finetti theorem implies the existence of  ``reference information''~$W$. The discrepancy between the two curves can then be understood as a discrepancy between the type of entropy measures that were calculated. While Hawking computed the (unconditional) von Neumann entropy $S(R)$, the recent gravitational calculations yield the conditional von Neumann entropy~$S(R|W)$. }
    \label{fig:page1}
\end{figure}

The (rather drastic) disagreement between Hawking's result and the Page curve is known as the \emph{black hole information puzzle}. It is crucial to note that the two conclusions derive from different assumptions, though. Hawkings' calculations are based on a rather explicit description of the radiation as a quantum field on curved spacetime. It is \emph{semi-classical} in that it is assumed that the spacetime geometry obeys the laws of (classical) general relativity. Conversely, the argument leading to the Page curve is entirely quantum-theoretical but does not take gravity into account --- the only input it takes from outside quantum theory is the mere fact that black holes do radiate.  The black hole information puzzle thus exhibits an (apparent) tension between gravity and quantum theory. 

Significant progress towards a resolution of the black hole information puzzle has been made recently~\cite{penington2020entanglement,almheiri2019entropy,almheiri2020page,penington2019replica,almheiri2020replica,almheiri2020entropy}. It has been shown that the Page curve can be reproduced by  explicit calculations using path integrals, which apart from the radiation field also take care of the gravitational degrees of freedom. The latter are approximated by saddle points that correspond to classical geometries. The approximation ensures that the calculations remain in the semi-classical regime and hence do not depend on speculations of how  a full theory of quantum gravity may look like.  

A curious feature of these novel calculations is that they yield an entropic quantity associated to the radiation field $R$ without ever telling us what the radiation state $\rho_R$ is. This is a general characteristic of the \emph{replica trick}~\cite{calabrese2004entanglement,calabrese2006entanglement,lewkowycz2013generalized,faulkner2013quantum}, on which the calculations are based. The replica trick relies on two general facts about entropies. Firstly, the von Neumann entropy of a quantum system equals the $n \to 1$ limit of the R\'enyi entropy of order $n>1$ (see Section~\ref{sec:replica} for a definition). Secondly, the R\'enyi entropy of order $n$, for integers $n \geq 2$,\footnote{Once the R\'enyi entropies of order $n$ are known for positive integers $n \in \mathbb{N}$, the values for non-integer orders are obtained by analytic continuation.} corresponds to the expectation value of an observable~$\tau_n$ on $n$ copies of the system, the \emph{replicas}. This expectation value can be evaluated by path integrals without an explicit description of the quantum state. The ingenious insight that led to the recent progress is that, when calculating the gravitational path integral for the observable $\tau_n$, there can be a dominant contribution from spacetime geometries consisting of wormholes that connect the $n$ replica black holes.

The discovery that a semi-classical approximation of gravitational path integrals is able to reproduce the Page curve, which was initially obtained by imposing unitarity  ``by hand'', is remarkable. It is even more so if one takes into account that the semi-classical regime is basically the same as the one considered  by Hawking,  who arrived however at the opposite conclusion of an ever increasing radiation entropy. From a purely mathematical viewpoint, the replica wormholes make all of the difference --- if one ignores them in the new calculations, one retrieves Hawking's result rather than the Page curve.  Conversely, in Hawking's original calculation~\cite{hawking1975particle,hawking1976breakdown,hartle1976path,gibbons1977action}, replica wormholes have no place, simply because his calculation doesn't rely on replicas. Hence, if we didn't have independent evidence for the Page curve,  it could well be that we would  have discarded the novel wormhole solutions as unphysical.

These considerations motivate a first conceptual question that we would like to address.

\begin{qn}\label{qn:1}
What is the physical reason for the discrepancy between the results for the radiation entropy as obtained by Hawking and via the replica trick, respectively? And which of the two is operationally meaningful?
\end{qn}

By ``operationally meaningful'' we mean that the quantity can (at least in principle) be determined by an experiment that acts on the radiation. The experiment may consist of applying quantum state tomography to estimate the density operator $\rho_R$ of the radiation field and then compute its von Neumann entropy. Since state tomography requires many identical copies of the same system, we would need an experimental procedure to prepare many identical black holes. Such an operational perspective has recently been suggested by Marolf and Maxfield (MM)~\cite{marolf2020transcending,marolf2020observations,marolf2021page} for the R\'enyi entropy of order~$n$ which, as noted above, can be interpreted as the expectation value of the observable~$\tau_n$.\footnote{To experimentally determine the expectation value of~$\tau_n$, one would need many identical copies of $n$-tuples of black holes. As a side remark, we note that there exist quantum algorithms that use such an approach to determine the entropy of a quantum system~\cite{troyer2017}.}

The main idea behind our approach to answering Question~\ref{qn:1} is to invoke a method that has its origin in quantum information theory: the \emph{quantum de Finetti theorem} (see Section~\ref{sec:dF} for a description). As we shall argue in Section~\ref{sec:results}, this theorem, applied to a many-black-hole experiment, implies the existence of particular information, $W$, which we call \emph{reference}. $W$~is not part of any single black hole nor its radiation field, but instead has a role analogous to a reference frame, which is required to make sense of the state of these systems (see Section~\ref{sec:elusivereference} for a discussion of this aspect). One may now associate two different entropic quantities to the black hole's radiation field. The unconditional von Neumann entropy, $S(R)$, measures one's uncertainty about the radiation when ignoring the reference~$W$, whereas the conditional von Neumann entropy, $S(R|W)$, measures this uncertainty when~$W$ is taken into account.

The answer to Question~\ref{qn:1} that is suggested by the de Finetti theorem is then as follows. Hawking calculated $S(R)$, whereas the novel calculations based on the replica trick correspond to a computation of $S(R|W)$. Furthermore, the latter is the entropy that we would find when carrying out a tomography experiment. These statements are direct consequences of Claim~\ref{claim:main1}, which is our first main result (see Section~\ref{sec:swapcond}). They also support and refine the general idea that Hawking actually calculated the entropy of an ensemble average of possible radiation states, $S(\langle \rho_R \rangle)$, whereas the replica trick calculations yield the ensemble average of the entropy of these radiation states, $\langle S(\rho_R) \rangle$~\cite{bousso2020gravity}.

Due to the discrepancy between Hawking's and Page's result, the mutual information between the reference and the radiation field, $I(R:W) = S(R) - S(R|W)$, can be large. In fact, since $S(R|W)$ equals zero after complete evaporation of the black hole, whereas $S(R)$ is at least as large as the initial black hole measured in terms of its Bekenstein-Hawking entropy, $S_{\mathrm{BH}}$, the information content of $W$ must be rather substantial.

In view of this, it may appear even more remarkable that the novel calculations, which are based on a semi-classical approximation to gravitational path integrals | rather than a fully quantum-mechanical argument | are able to determine $S(R|W)$, the entropy conditioned on~$W$. After all, the calculations require no input from gravity other than (classical) saddle point geometries. Their uncanny efficacy is thus somewhat unsettling. This brings us to a second conceptual puzzle that we would like to address. 

\begin{qn}\label{qn:2}
  How is the reference~$W$ represented physically in spacetime, and why can a semi-classical gravitational path integral ``know'' enough about~$W$ to yield $S(R|W)$?
\end{qn}

To answer this question, we once again invoke the quantum de Finetti theorem. It yields a second main result, Claim~\ref{claim:main2}, which asserts that the conditional entropy $S(R|W)$ is equal to the normalised total entropy $S(R_1 \cdots R_N)/N$ of the radiation fields of a collection of $N$  black holes, prepared identically in the same spacetime, in the limit of large~$N$ (see Section~\ref{sec:condcorr}). Thus, clearly, if $S(R|W)$ is smaller than $S(R)$, then $S(R_1 \cdots R_N)$ is smaller than $S(R_1) + \cdots + S(R_N)$, which in turn means that the radiation fields of the different black holes are correlated. In other words, if $W$ is non-trivial in that $S(R|W) < S(R)$, then it manifests itself as correlation between the different black hole systems. This answers the first part of Question~\ref{qn:2}. 

In the semi-classical gravitational path integral calculations, geometries that feature wormholes between replicas become relevant precisely in the regime where the reference~$W$ becomes non-trivial in the sense described above. The wormholes may thus be interpreted as a geometrical manifestation of the correlation between the black hole systems that is implied by a non-trivial~$W$. This is in agreement with recent results by MM, who studied the baby universe that forms at the common interior of replicas of black holes connected by wormholes. Our reference~$W$ would in their model be  encoded into different possible states, the $\alpha$-states, of the baby universe Hilbert space~\cite{marolf2020transcending,marolf2020observations,marolf2021page} (see also~\cite{giddings2020wormhole}).\footnote{Our results thus indicate that MM's conclusions hold more generally and may not need to be based on the premise of baby universes.} These considerations suggest the following answer to the second part of Question~\ref{qn:2}. It is the replica wormholes that ``tell'' the path integral about $W$ and thus enable the computation of~$S(R|W)$.

The remainder of this article is structured as follows. In Section~\ref{sec:prelim} we provide introductory remarks on the gravitational path integral calculations using the replica trick, on the quantum de Finetti theorem, and on the notion of conditional entropy. Section~\ref{sec:results} contains our main claims, which address the two questions posed above, and their proofs, which are largely based on the quantum de Finetti theorem. In Section~\ref{sec:model} we propose an extension of the random unitary model for black holes, which is suggested by the quantum de Finetti theorem.  We conclude in  Section~\ref{sec:discussion}, where we also address criticism of the recent path integral calculations in the light of our results. 

\section{Preliminaries}\label{sec:prelim}

\subsection{The replica trick calculation in gravity}\label{sec:replica}

We start with a brief review of the calculation of the radiation entropy based on the replica trick~\cite{penington2019replica,almheiri2020replica}, intended for readers who are not familiar with the recent developments.\footnote{For more background and details we recommend the review~\cite{almheiri2020entropy}.} For simplicity, the calculations are often carried out under the assumption that spacetime has a Euclidean rather than a Lorentzian signature, and we thus also focus here on the Euclidean case.\footnote{We refer to~\cite{marolf2020observations} for a more comprehensive Lorentzian treatment.}

The following considerations are quite general and  apply to basically any quantum system. Nonetheless, for our purposes the system will typically be a spacetime containing one or several black holes. We thus assume that the degrees of freedom consist of the spacetime geometry $g$ (we take $g$ to be a description of the topology and the metric) as well as a matter field $\psi$ (in the following called the \emph{quantum field}) that lives on the spacetime. These degrees of freedom shall be determined by boundary conditions $\mathcal{B}$ on a boundary region that is situated at a large distance from any black holes.  Then, according to the Feynman path integral prescription, the expectation value of an observable $O = O[g, \psi]$ is given by the expression
\begin{align} \label{eq:observable}
    \langle O \rangle_{\mathcal{B}} = \frac{Z[\mathcal{B}, O]}{Z[\mathcal{B}]}
\end{align}
where
\begin{align} \label{eq:partitionobservable}
    Z[\mathcal{B}, O] = \int_{\mathcal{B}}\mathcal{D}g\mathcal{D}\psi O[g, \psi] e^{-I[g,\psi]}
\end{align}
is the partition function with the observable $O$ inserted, and $Z[\mathcal{B}] = Z[\mathcal{B}, 1]$ is the ``plain'' partition function, which serves as a normalisation in~\eqref{eq:observable}. The integral above runs over all possible configurations of~$g$ and $\psi$ that are compatible with the boundary conditions~$\mathcal{B}$, and $I$ is the (Euclidean) action. The latter is typically assumed to consist of the Einstein-Hilbert action (or a variant thereof; see Section~\ref{sec:pssy} for an example) that is proportional to the inverse of the gravitational constant~$G_N$ and a term for the quantum field $\psi$ with minimal coupling to gravity. 

The replica trick is a method that enables the computation of the entropy of a subsystem~$R$, in our case the field of Hawking radiation emitted by a black hole defined by boundary conditions~$\mathcal{B}$, using an expression of the form~\eqref{eq:observable}. As noted in the introduction, it relies on the fact that the von Neumann entropy of a quantum system $R$ in state $\rho_R$ can be written as $S(R)_{\rho} = \lim_{n \to 1} S_n(R)_{\rho}$, where $S_n(R)_{\rho} = \frac{1}{1-n} \log \tr(\rho_R^n)$ are the R\'enyi entropies of order~$n$.  For $n \in \mathbb{N}$, the trace occurring in the latter can be expressed as
\begin{align} \label{eq:replicabasic}
  \tr(\rho_R^n) =  \tr(\rho_R^{\otimes n} \tau_n) 
\end{align}
where $\tau_n = \tau_{R_1 \cdots R_n}$ is the cyclic shift operator, which moves the content of $R_i$ to $R_{i+1 \operatorname{mod} n}$ for any $i = 1, \ldots, n$. Since the expression on the right hand side corresponds to an expectation value of $\tau_n$, interpreted as an observable on $n$ identical copies of the original system of interest, one may invoke expression~\eqref{eq:observable}.

The recent calculations~\cite{penington2019replica,almheiri2020replica} of the radiation entropy are motivated by this idea. However, instead of $n$ identical copies of the black hole spacetime, the calculations refer to a system defined by $n$ copies of the boundary conditions $\mathcal{B}$. As we shall see, this makes an important difference. To keep track of this difference, we will call the resulting entropic quantities \emph{swap entropies} --- a term borrowed from MM~\cite{marolf2020observations}. Specifically, and in analogy to the R\'enyi entropy of order~$n$, we define the \emph{swap entropy} of order $n$ by\footnote{Following the notation used for entropies defined by states $\rho$, we put a subscript to indicate the boundary conditions $\mathcal{B}$ which define the state of the system.}
\begin{align}\label{eq:replica_trick_1}
  S_n^{\mathrm{swap}}(R)_{\mathcal{B}} = \frac{1}{1-n} \log \langle \tau_n \rangle_{\mathcal{B}^{\times n}} \ .
\end{align}
While the results derived in this work rely on definition~\eqref{eq:replica_trick_1} for the swap entropy, we note that in the literature this expression is sometimes replaced by  
\begin{align}\label{eq:replica_trick_2}
  S_n^{\mathrm{swap}}(R)_{\mathcal{B}} = \frac{1}{1-n} \log \frac{Z_n}{Z_1^n} \ ,
\end{align}
where $Z_n := Z[\mathcal{B}^{\times n}, \tau_n]$. This expression is identical to~\eqref{eq:replica_trick_1} under the assumption that $Z[\mathcal{B}^{\times n}] = Z[\mathcal{B}]^n$.\footnote{The assumption is equivalent to say that, when $R$ is trivial, the dominant contribution for the swap entropy comes only from a disconnected geometry. This can be reasonably justified in simple models~\cite{penington2019replica}. However, in the case where replica wormhole contributions become relevant in $Z[\mathcal{B}^{\times n}]$, the assumption may fail,  as pointed out in~\cite{engelhardt2021free}. We then have to resort to \eqref{eq:replica_trick_1} instead to avoid pathological results such as a positive entropy for a pure state.}

In Euclidean signature, the boundary conditions $\mathcal{B}$ are usually specified on the asymptotic region of the spacetime that has a flat geometry with topology $\mathcal{N} \times \mathcal{I}$, where $\mathcal{N}$ is a spatial region at an asymptotically large distance from the black hole,\footnote{Note that here the boundary conditions are prescribed on the boundary region rather than the topological boundary, so $\mathcal{N} \times \mathcal{I}$ has the same dimension as the spacetime that satisfies the boundary conditions. $\mathcal{N}$ is often taken to be compact, corresponding to an IR cutoff.} and where $\mathcal{I}$ is a circle in the imaginary time direction of length (inverse temperature) $\beta$.  The subsystem containing the radiation $R$ can be defined as a cut on $\mathcal{N}$.\footnote{In Lorentzian signature, one considers a Schwinger-Keldysh contour with $n$ identical and independent copies of the past boundary conditions, and the future boundaries at~$R$ are left open~\cite{marolf2020observations}. } For the computation of $Z_n = Z[\mathcal{B}^{\times n}, \tau_n]$, the operator $\tau_n$, which connects the radiation field of the $i$th black hole to the $(i+1)$st, can then be regarded as a part of the boundary conditions, as shown by Fig.~\ref{fig:bc}. 

\begin{figure}
     \centering
     \begin{subfigure}[b]{0.13\textwidth}
         \centering
         \includegraphics[width=\textwidth]{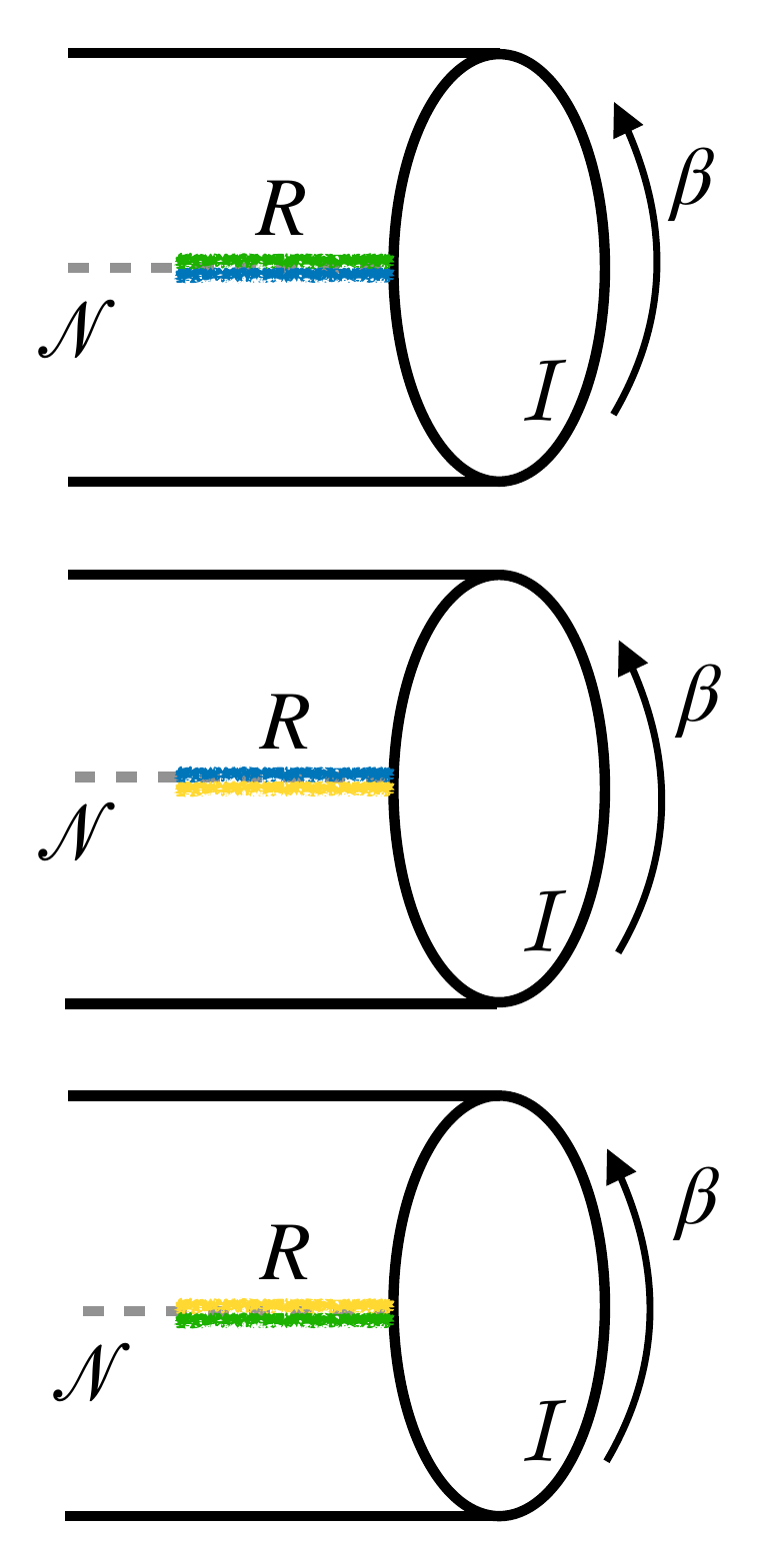}
         \caption{$\mathcal{B}^{\times n}, \tau_n$.}
         \label{fig:bc}
     \end{subfigure}
     \hfill
     \begin{subfigure}[b]{0.45\textwidth}
         \centering
         \includegraphics[width=\textwidth]{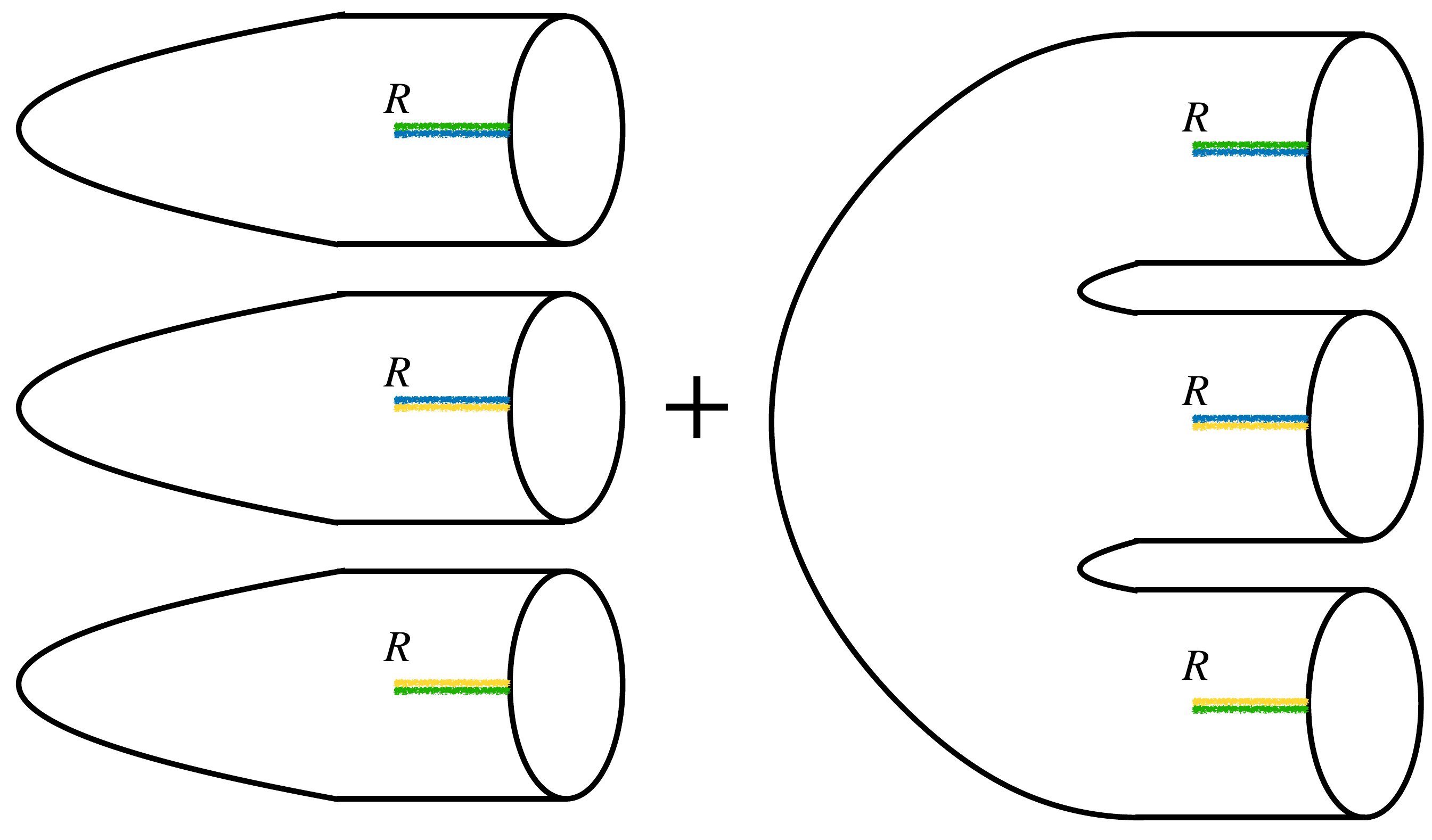}
         \caption{The two $\mathbb{Z}_n$-symmetric saddles.}
         \label{fig:saddles}
     \end{subfigure}
     \hfill
     \begin{subfigure}[b]{0.3\textwidth}
         \centering
         \includegraphics[width=0.65\textwidth]{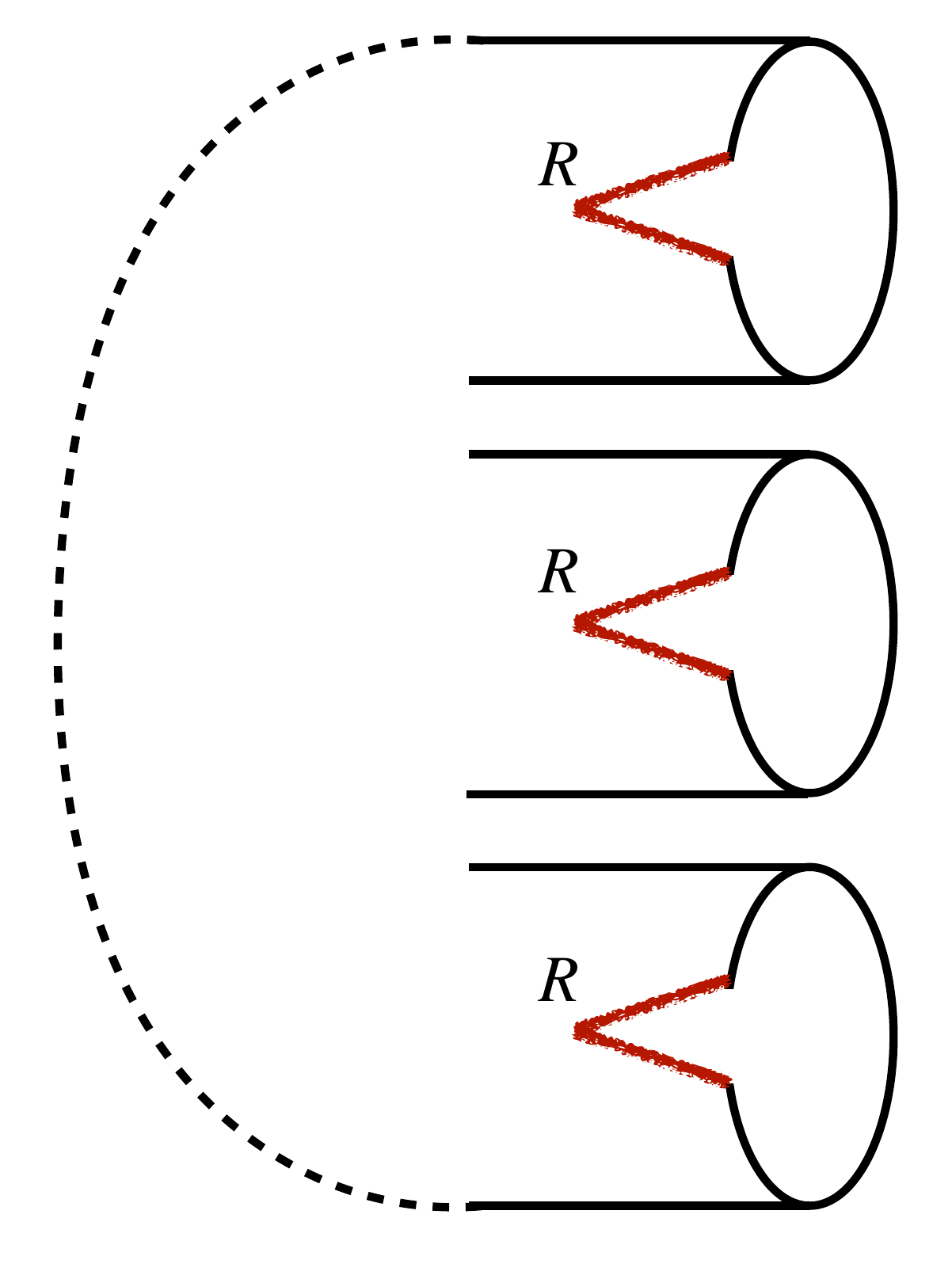}
         \caption{The density operator.}
         \label{fig:densityoperator}
     \end{subfigure}
        \caption{Fig.~\ref{fig:bc} illustrates the boundary conditions for the replica trick path integral. The manifold of the boundary region of each replica, where the conditions $\mathcal{B}$ are imposed, factorises into a spatial part $\mathcal{N}$ (grey dashed line) and a periodic Euclidean time part $\mathcal{I}$ (black circle). The radiation system~$R$ is a cut on $\mathcal{N}$. The cyclic shift operator $\tau_n$ may then be included in the boundary conditions, where it imposes a cyclic gluing of the cuts $R$ of different replica, as indicated by the colours. Fig.~\ref{fig:saddles}  shows the two dominant geometries that contribute to the path integral under the assumption of $\mathbb{Z}_n$ symmetry. These are the disconnected Hawking saddle (left) and the connected replica wormholes (right). Fig.~\ref{fig:densityoperator}  shows the density operator of the radiation $\rho_{R_1\cdots R_n}$ defined by the path integral cut open at the regions $R$, which are highlighted in red. The dashed line indicates the summation over all possible geometries.}
        \label{fig:replica}
\end{figure}

Up to this point we have implicitly assumed that the path integral over the geometries $g$ is well defined and can in principle be calculated. However, lacking a full theory of quantum gravity, we do not know whether this is the case. The standard approach to circumvent this problem is to resort to a semi-classical regime, which corresponds to the limit where the gravitational coupling $G_N$ is small. Specifically,  the integration over $g$ is replaced by a saddle-point approximation and it is assumed that the dominant contributions to the partition function $Z_n$ come from geometries $g$ that are classical solutions of the gravity action.

Under the additional assumption, which is usually made within the replica-trick based calculations~\cite{lewkowycz2013generalized}, that the dominant classical geometry $g$ respects the $\mathbb{Z}_n$ symmetry\footnote{\label{ftn:repsym}The $\mathbb{Z}_n$ replica symmetry may be broken when we have a non-trivial initial state that forms the black hole~\cite{akers2021leading,wang2021refined} (see also \cite{engelhardt2021free}). Then we generally cannot discard the contributions from other saddles, and including them changes the Page curve as plotted in Fig.~\ref{fig:page1}, particularly near the transition. One such example is shown in Fig.~\ref{fig:page2} where the Page curve for a black hole in superposition of two evaporation stages is given.} of the boundary conditions imposed by the cyclic shift operator $\tau_n$, there exist two candidates for potentially dominant saddles, one referred to as the \emph{Hawking saddle} and the other as the \emph{replica wormholes}. They are illustrated in Fig.~\ref{fig:saddles}. The Hawking saddle consists of disconnected geometries that correspond to identical copies of a (Euclidean Schwarzschild) black hole, whereas the replica wormholes geometry consists of connections between the different replicas. 

The argument so far provides a recipe for computing $S_n^{\mathrm{swap}}(R)_{\mathcal{B}}$ for $n \in \mathbb{N}$. Following the replica trick, we need to take the $n \to 1$ limit, i.e., determine\footnote{In the standard literature, this quantity is usually just denoted $S(R)$.  Here we keep the swap entropy terminology to remind ourselves that the quantity is obtained via the replica trick.}
\begin{align} \label{eq:swaplimit}
  S^{\mathrm{swap}}(R)_{\mathcal{B}} = \lim_{n \to 1} S_n^{\mathrm{swap}}(R)_{\mathcal{B}} \ . 
\end{align}
This is achieved by considering the analytic continuation of the expression for $S_n^{\mathrm{swap}}(R)_{\mathcal{B}}$. Since a geometry $g_n$ with boundary conditions $\mathcal{B}^{\times n}$ for non-integer $n$ is difficult to define, one first maps $g_n$ to a quotient geometry $\hat{g}_n$ that has boundary conditions $\mathcal{B}$ corresponding to a single black hole, but takes into account the replica wormholes by adding appropriate terms to the gravitational part of the action, $I[\hat{g}_n]$. The resulting path integral can then be continued analytically to arbitrary values~$n$, and the limit~\eqref{eq:swaplimit} can be computed.

The geometry that one obtains in this limit is basically that of a single black hole. However, the replica wormholes leave a trace on the black hole metric, which is geometrically manifested as an additional region $I$, called the \emph{island}. Eventually, taking into account the two saddles, one obtains an expression that is known as the \emph{island formula}~\cite{penington2019replica,almheiri2020page,almheiri2020replica},\footnote{The accuracy of the approximation depends on the degree to which the assumptions hold, e.g., that the contributions of solutions which break the $\mathbb{Z}_n$ replica symmetry can be neglected (see Footnote~\ref{ftn:repsym} above).} 
\begin{equation}\label{eq:island}
    S^{\mathrm{swap}}(R)_{\mathcal{B}}
    \approx \mathrm{min\,ext}_I\left[\frac{A[\partial I]}{4G_N}+S(R\cup I)_{\rho_{\psi}} \right]\,.
\end{equation}
The expression in the square brackets is a \emph{generalised entropy}~\cite{Bekenstein1973}, defined as the sum of the area $A[\partial I]$ of the boundary of the island $I$ and the von Neumann entropy $S(R\cup I)_{\rho_{\psi}}$ of the matter field $\psi$ (at inverse temperature $\beta$) in the joint region consisting of $R$ and $I$. One then varies over islands~$I$ and takes the minimum over those $I$ where the local variation of the generalised entropy vanishes.\footnote{The island formula is closely related to the Ryu-Takayanagi (RT) formula in AdS/CFT~\cite{ryu2006holographic,ryu2006aspects}, which computes the entanglement entropy of the reduced state on a boundary subregion using the area of a minimal surface in the bulk. The quantum RT formula was later proposed to account for the bulk entropy contribution~\cite{faulkner2013quantum}. It then led to the quantum extremal surface prescription~\cite{engelhardt2015quantum}, which was used to obtain the Page curve~\cite{penington2020entanglement,almheiri2019entropy} before the island formula was proposed shortly after. In fact, the island formula exactly follows from the quantum extremal surface prescription in doubly holographic models~\cite{almheiri2020page,chen2021evaporating,chen2020quantum1}.} Remarkably, $S^{\mathrm{swap}}(R)_{\mathcal{B}}$ reproduces the Page curve, which is shown in Fig.~\ref{fig:page1}, to good approximation. 

Hawking's original calculations~\cite{hawking1975particle,hawking1976breakdown,hartle1976path,gibbons1977action} use the same semi-classical approximation as the replica-trick based ones that we just described. In fact, the former can be retrieved from the latter if one ignores the contributions from the replica wormholes. More precisely, if one keeps only the Hawking saddle in the calculation described above, one obtains an expression that corresponds to~\eqref{eq:island} with an empty island~$I$, that is, just $S(R)_{\rho_{\psi}}$, where $\rho_{\psi}$ is the Hartle-Hawking state of the radiation field as in  Hawking's original calculation. It thus looks as if Hawking missed the replica wormhole solution, which is responsible for bringing the entropy down to the Page curve. Note however that Hawking's approach to compute the entropy does not require replicas in the first place, i.e., replica wormholes have no place there.  

To better understand the quantity $S^{\mathrm{swap}}(R)_{\mathcal{B}}$ obtained via the replica trick, it is worth noting that the very existence of the replica wormholes indicates that the state of the $n$ radiation subsystems $\rho_{R_1 \cdots R_n}$ defined implicitly by the boundary conditions $\mathcal{B}^{\times n}$ (c.f.\ Fig.~\ref{fig:densityoperator})  does not factorise into a product $\rho_{R_1} \otimes \cdots \otimes \rho_{R_n}$ --- despite the fact that the boundary conditions have such a product structure. This is a curious feature of  gravitational path integrals, known as the \emph{factorisation problem}~\cite{giddings1988loss,coleman1988black,polchinski1994possible,maldacena2004wormholes,arkani2007}.

The factorisation problem gives rise to the \emph{state paradox}  pointed out in~\cite{bousso2020gravity}. Recall that the original replica trick is based on~\eqref{eq:replicabasic}, i.e., the entropy is supposed to be derived from  observables $\tau_n$ evaluated on an $n$-fold product of the radiation state, $\rho_{R}^{\otimes n}$. Conversely, due to the factorisation problem, $S^{\mathrm{swap}}(R)_{\mathcal{B}}$ actually derives from evaluating $\tau_n$ on states $\rho_{R_1 \cdots R_n}$ that are not in general of product form and hence different from $\rho_{R}^{\otimes n}$. We shall resolve this tension in Section~\ref{sec:results}.

\subsection{The quantum de Finetti theorem} \label{sec:dF}

The quantum de Finetti theorem~\cite{hudson1976locally,caves2002unknown,renner2007symmetry,renner2008security,christandl2007one,konig2005finetti,koenig2009most} is a generalisation of a theorem from classical statistics, which is widely used in quantum information theory~\cite{fannes2006finite,brandao2011faithful,christandl2012reliable,brandao2017quantum,lancien2017flexible,harrow2013church,arnon2015finetti,arnon2016non}. Roughly speaking, it asserts that, if we choose $n$ subsystems $R_1, \ldots, R_n$ at random from an $N$-partite system, where $N \gg n$,  then the joint state of $R_1, \ldots, R_n$ is well approximated by a convex combination of product states of the form~$\sigma_R^{\otimes n}$. 

To phrase this more precisely, we need to introduce a few definitions. Let $\mathcal{S}(\h_R)$ denote the set of density operators on a Hilbert space $\h_R$. A density operator on $n$ identical systems $R_1, \ldots, R_n$ and an extra system $E$,  $\rho_{R_1 \cdots R_n E}\in\mathcal{S}(\h_R^{\otimes n} \otimes \h_E)$, is said to be \emph{permutation-invariant on $R_1, \ldots, R_n$ relative to $E$} if it transforms trivially under the action of the symmetric group~$\mathbf{S}_n$, i.e., 
\begin{equation}
    \rho_{R_1 \cdots R_n E}= (U_\pi \otimes I_E) \rho_{R_1 \cdots R_n E} (U_\pi^\dagger \otimes I_E) \quad\forall\pi\in \mathbf{S}_n \ ,
\end{equation}
where $U_\pi$ denotes the unitary representation of $\pi$ on $\h_R^{\otimes n}$. Furthermore, $\rho_{R^n E}$ is called \emph{$N$-exchangeable on $R_1, \ldots, R_n$ relative to $E$} if there exists an extension $\rho_{R_1 \cdots R_N E} \in\mathcal{S}(\h_R^{\otimes {N}} \otimes \h_E)$ that includes $m = N-n$ additional systems, i.e., 
\begin{equation}\label{eq:exchangeable}
    \tr_m \rho_{R_1 \cdots R_N E} = \rho_{R^{n} E} \ ,
\end{equation}
which is permutation-invariant on $R_1, \ldots, R_N$ relative to~$E$. A density operator that is $N$-exchangeable on $R_1, \ldots, R_n$ for any $N\geq n$ is said to be \emph{infinitely exchangeable}, or just \emph{exchangeable}, on $R_1, \ldots, R_n$. To relate these definitions to the introductory paragraph, note that the density operator of $n$ subsystems chosen at random from $N$ subsystems is $N$-exchangeable. Finally, $\rho_{R^n E}$ is said to be of \emph{de Finetti form on $R_1, \ldots, R_n$ relative to $E$} if there exists a probability measure $\dd \sigma$ on $\mathcal{S}(\h_R)$ and a family $\{\rho_{E|\sigma}\}_\sigma$ of density operators on $\h_E$ parametrised by $\sigma \in \mathcal{S}(\h_R)$ such that   
\begin{align} \label{eq:dFstate}
    \rho_{R_1 \cdots R_n E} = \int\dd  \sigma \sigma_{R}^{\otimes n} \otimes \rho_{E|\sigma}\ .
\end{align}

In its basic form, the quantum de Finetti theorem asserts that any density operator that is infinitely exchangeable on finite-dimensional systems $R_1,  \ldots, R_n$ is of de Finetti form on these systems~\cite{hudson1976locally,caves2002unknown,konig2005finetti}.\footnote{\label{ftn:uniqueness}The measure $\dd \sigma$ in~\eqref{eq:dFstate} is asymptotically unique, in the sense that for any $\dd \tilde{\sigma} \neq \dd\sigma$ there exists $N \in \mathbb{N}$ such that $\int\dd \tilde{\sigma} \tilde{\sigma}_{R}^{\otimes N} \neq \int\dd  \sigma \sigma_{R}^{\otimes N}$. The de Finetti theorem therefore implies that for any family $
\{\rho_{R^n} \in\mathcal{S}(\h_R^{\otimes n})\}_{n \in \mathbb{N}}$ of  permutation-invariant density operators that are mutually compatible, i.e., $\tr_m \rho_{R^{n+m}} = \rho_{R^{n}} \, \forall n,m \in \mathbb{N}$, there exists a unique measure $\dd \sigma$ such that~\eqref{eq:dFstate} holds.} For our purposes, we need a robust version of this statement, which yields a claim when the state is $N$-exchangeable for some finite $N \in \mathbb{N}$ only. 

\begin{thm}\label{thm:definetti_state}
For any even $N$ and for any density operator $\rho_{R_1 \cdots R_{N/2} E}$ that is $N$-exchangeable on $R_1 \cdots R_{N/2}$ there exists an extension with a classical variable\footnote{To describe a classical variable~$W$ on a set~$\mathcal{W}$ in the context of quantum systems, we encode it in a quantum system, which we also call~$W$,  equipped with a basis $\{\ket{w}_W\}_{w \in \mathcal{W}}$ labelled by the elements of~$\mathcal{W}$.} $W$, which takes values in a finite set~$\mathcal{W}$, such that, for any $n \leq N/2$,
  \begin{align} \label{eq:dFrobust}
      \bigl\| \rho_{R_1 \cdots R_n E W} - \sum_w p_w \rho_{R|w}^{\otimes n} \otimes \rho_{E|w} \otimes \ketbra{w}{w}_W \bigr\|_1 \leq
        6\dim(R)\sqrt{\frac{2 n }{N}} \ ,
  \end{align}
  where the probability distribution $p$ and the density operators $\rho_{R|w}$ and $\rho_{E|w}$ are defined by  $\sum_{w} {p_w \rho_{R|w} \otimes \ketbra{w}{w}_W} = \rho_{R_1 W}$ and $\sum_{w} {p_w \rho_{E|w} \otimes \ketbra{w}{w}_W} = \rho_{E W}$.
\end{thm}
\begin{proof}
  See Appendix~\ref{app:deFinetti}. 
\end{proof}
This version of the quantum de Finetti theorem can be obtained along the lines of the proofs in~\cite{christandl2007one,renner2007symmetry,renner2008security}. It is similar to Theorem~II.7$'$ in~\cite{christandl2007one}, although the right hand side involves an extra square root. Conversely, it generalises the latter in four different aspects. First, it asserts that the distribution $p$ in the de Finetti state is independent of $n$. Second, it includes the classical extension $W$. Third, it asserts that the integral in~\eqref{eq:dFstate} can be replaced by a finite sum. Fourth, it yields exact statements for $n=1$. We also note that  $W$ may be obtained by an appropriate measurement applied to the subsystems $R_{N/2+1}, \ldots, R_N$ of a permutation-invariant extension of $\rho_{R_1 \cdots R_{N/2} E}$, as can be verified by inspecting the proof.

\subsection{Conditional entropy}

Given a density operator $\rho = \rho_{A B}$ on a joint system consisting of parts $A$ and $B$, the von Neumann entropy of~$A$ conditioned on~$B$ is defined as 
\begin{align}
 S(A|B)_{\rho} = S(A B)_{\rho} - S(B)_{\rho} \ .
\end{align} 
We will be interested in the case where the conditioning system $B$ is classical, i.e., $\rho_{A B} = \sum_{b} p_b \rho_{A|b} \otimes \ketbra{b}{b}_B$, with $\rho_{A|b}$ the state of~$A$ conditioned on ${B=b}$. In this case, we may define $S(A|B=b)_{\rho} := S(A)_{\rho_{|b}}$ and express the conditional entropy as the average of this value over~$B$,
\begin{align} \label{eq:qcconditionalvN}
     S(A|B)_{\rho} = \sum_{b} p_b S(A|B=b)_{\rho} \ .
\end{align}

The notion of conditional entropies can be extended  to other entropic quantities~\cite{renner2008security}. Specifically, the R\'enyi entropy of order $n$ of $A$ conditioned on~$B$ may be defined as\footnote{We note that there exist other variants of conditional R\'enyi entropy in the literature. For example, the definition used in~\cite{muller2013quantum} has the reduced state $\rho_B$ replaced by the maximum over arbitrary density operators $\sigma_B$. These definitions also satisfy~\eqref{eq:Renyilimit}.} 
\begin{align}
      S_{n}(A|B)_\rho = - D_n({\rho_{AB} \| I_A \otimes \rho_B}) \ ,
\end{align}
where $D_n(\cdot \| \cdot)$ denotes the sandwiched relative entropy~\cite{wilde2014strong,muller2013quantum}. It follows immediately from Theorem~3 of~\cite{muller2013quantum} that
\begin{align} \label{eq:Renyilimit}
  S(A|B)_\rho = \lim_{n \to 1}  S_{n}(A|B)_{\rho} \ .
\end{align}
In the case where the conditioning system $B$ is classical, \eqref{eq:qcconditionalvN} generalises to 
\begin{align} \label{eq:qcconditionalRenyi}
  S_{n}(A|B)_\rho 
  = \frac{1}{1-n} \log \sum_b p_b \tr(\rho_{A|b}^n) 
  = \frac{1}{1-n} \log \sum_b p_b 2^{(1-n) S_n(A|B=b)_{\rho}} \ ,
\end{align}
where $S_n(A|B=b)_\rho := S_n(A)_{\rho_{|b}}$, as can be readily verified from the definition of $D_n(\cdot \| \cdot)$.

\section{Demystifying the replica trick} \label{sec:results}

\subsection{A many-black-hole system} \label{sec:manyblackholes}

Consider a setup consisting of $N$ black holes, for any fixed $N \in \mathbb{N}$, each specified by the same boundary conditions $\mathcal{B}$. While these boundary conditions may be chosen arbitrarily, and in particular, do not need to lie in the past, we will for concreteness think of a mechanism that prepares $N$ identical matter shells, $M_1, \ldots, M_N$, at locations chosen at random within a given spatial region.  The region shall be large enough so that the individual shells are separated by a large spatial distance (with probability close to~$1$). This ensures that we can neglect any interaction among them that could be mediated by the space in between. We reflect this separation formally by writing the overall boundary conditions as a Cartesian product, $\mathcal{B}^{\times N}$. 

By definition, the mechanism prepares all matter shells $M_1, \ldots, M_N$ in the same way so that, except for their different locations, they are indistinguishable. Operationally, this means that there does not exist any experiment to tell whether or not an extra \emph{swap operation} was applied, which exchanged two of the shells. This indistinguishability property is clearly preserved as the matter shells collapse to black holes and start to emit radiation. (Otherwise, letting them collapse would precisely be an experiment to distinguish them, which is excluded by definition.)

Due to the large spatial separation between the collapsing matter shells, we may assume that the resulting radiation fields, $R_1, \ldots, R_N$, can be treated as separate subsystems. Indistinguishability as discussed above then implies that the joint state $\rho_{R_1 \cdots R_N}$ of all radiation systems at a given ``time'' is invariant under a swap operation that exchanges two of the radiation fields.\footnote{For our purposes it is irrelevant how precisely such a time is defined, provided that the definition is invariant under the swap operation. For example, one may imagine that the region in which the initial matter shells are placed is itself a huge shell with an observer sitting at the center of it. Time may then be defined with respect to this observer.} The state is thus permutation-invariant on $R_1, \ldots, R_N$.\footnote{Even if the different black holes were generated in different ways, one may enforce permutation invariance operationally by shuffling them at random, i.e., by applying the swap operation to randomly chosen pairs. In fact, applying such a reshuffling to achieve permutation invariance is a common trick in quantum information theory~\cite{renner2007symmetry}.} 


\subsection{Non-signalling property for path integrals} \label{sec:nonsignalling}

In the description above we have made the assumption that $R_1, \ldots, R_N$ can be treated as different subsystems. The non-signalling property of quantum theory then implies that the value of an observable $O_n$ defined on $n$ out of the $N$ subsystems does not depend on the remaining $m=N-n$ subsystems. That is, instead of evaluating $O_n$  on $\rho_{R_1 \cdots R_N}$, we may equally well evaluate it on the reduced state $\rho_{R_1 \cdots R_n} = \tr_{m}(\rho_{R_1 \cdots R_{n+m}})$, i.e., 
\begin{equation}\label{eq:nonsignalling}
   \tr(\rho_{R_1 \cdots R_{n}}O_n) = \tr(\rho_{R_1 \cdots R_{n+m}} \cdot O_n \otimes I^{\otimes m}) \, .
\end{equation}
In fact, it is precisely this non-signalling property that justifies the use of reduced states in quantum theory. 

In the path integral formalism, states are defined implicitly via boundary conditions. If the boundary conditions have product form, one may take this to mean that they refer to different subsystems. The non-signalling property~\eqref{eq:nonsignalling} would then imply that the value of an observable $O_n$ on $n$ subsystems should not depend on whether the boundary conditions prepare additional subsystems, i.e.,
\begin{align} \label{eq:pinonsignalling}
    \langle O_n \rangle_{\mathcal{B}^{\times n}} 
    =     \frac{\int_{\mathcal{B}^{\times (n+m)}}\mathcal{D}g\mathcal{D}\psi\, O_n e^{-I[g,\psi]}}{\int_{\mathcal{B}^{\times (n+m)}}\mathcal{D}g\mathcal{D}\psi\, e^{-I[g,\psi]}} \, .
\end{align}

Note however that there is a subtle conceptual difference between~\eqref{eq:nonsignalling} and~\eqref{eq:pinonsignalling}. In the first, the reduced state $\rho_{R_1 \cdots R_n} = \tr_{N-n}(\rho_{R_1 \cdots R_{N}})$ that appears on the left hand side describes $n$ radiation systems within a setup consisting of $N$ black holes. Conversely, the expectation value on the left hand side of~\eqref{eq:pinonsignalling} may be understood as an $N$-independent quantity that refers to an experiment where only $n$ black holes are prepared in the first place. But this means that~\eqref{eq:pinonsignalling}~can only be valid for all $m \in \mathbb{N}$ under the additional assumption that an experiment on the radiation fields of $n$ black holes does not depend on the presence of additional $m=N-n$ black holes. 

Since factorisable boundary conditions do not necessarily yield factorisable states (this is precisely the factorisation problem), it is generally challenging to verify~\eqref{eq:pinonsignalling} by a direct calculation.\footnote{Non-factorisability of the state does however not contradict non-signalling.} The validity of~\eqref{eq:pinonsignalling} is thus often imposed by assumption.\footnote{In the context of path integral calculations based on the replica trick, one usually makes the additional assumption that $\int_{\mathcal{B}^{\times n}}\mathcal{D}g\mathcal{D}\psi\, e^{-I[g,\psi]} = (\int_{\mathcal{B}}\mathcal{D}g\mathcal{D}\psi\, e^{-I[g,\psi]})^n $; see the discussion around~\eqref{eq:replica_trick_2}.} For the purpose of this work, it suffices however to assume that this condition holds asymptotically for large~$m$, in the sense that the expression on the right hand side of~\eqref{eq:pinonsignalling} converges to a well-defined value in the limit $m \to \infty$. We then take this to be the definition of the expectation value on the left hand side.

\subsection{Definition of the reference~$W$} \label{sec:reference}

Suppose that we generate an $N$-black-hole system with radiation subsystems $R_1, \ldots, R_N$, by imposing boundary conditions~$\mathcal{B}^{\times N}$ as described in Section~\ref{sec:manyblackholes}, for any even integer~$N$. The quantum de Finetti theorem, Theorem~\ref{thm:definetti_state},\footnote{The robustness of the de Finetti theorem, i.e., the approximate validity of the statement for finite~$N$, is important for this construction to have an operational meaning. Note that any given procedure for generating black holes may satisfy our requirements (e.g., that they are separated by large spatial distances and hence approximately have a subsystem structure) only up to a finite number $N$. We may however still safely assume that, given any arbitrary $N$, there exists a procedure that generates $N$ black holes that meet our requirements.} asserts the existence of a quantum-classical extension $\rho_{R_1 \cdots R_{N/2} W}$ of the state of $N/2$ radiation systems that includes an extra system~$W$, the reference, such that the following holds. There exists a de Finetti approximation of the form\footnote{The superscript $(N)$ shall  remind us that the state is an approximation that depends on $N$.}
\begin{align} \label{eq:qcstate}
  \rho^{(N)}_{R_1 \cdots R_{N/2} W} = \sum_w\, p^{(N)}_w \rho_{R|w}^{\otimes N/2} \otimes\ketbra{w}{w}_W \ ,
\end{align}
such that the reduced states on $n \leq N/2$ radiation systems satisfy 
\begin{align} \label{eq:dFapprox}
  \| \rho_{R_1 \cdots R_n W} - \rho^{(N)}_{R_1 \cdots R_n W} \| \leq \O(\sqrt{n/N})
\end{align}
and, for any $i \in \{1, \ldots, N/2\}$,
\begin{align} \label{eq:dFprecise}
  \rho_{R_i W} = \rho^{(N)}_{R_i W} \ .
\end{align}

Note that, if we knew the value $w$ that $W$ admits, we could say that any individual radiation system $R=R_i$ is in state $\rho_{R|w}$. In particular, conditioned on ${W=w}$, any of them has  entropy $S(R|{W=w})_{\rho^{(N)}} = S(R)_{\rho_{|w}}$. According to~\eqref{eq:qcconditionalvN}, the expectation of this entropy over all possible values of $W$ is equal to the conditional entropy, $S(R|W)_{\rho^{(N)}}$. Note furthermore that, due to~\eqref{eq:dFprecise}, this value may as well be understood as an entropy of the quantum-classical extension of the original state rather than its de Finetti approximation.

The conditional entropy $S(R|W)_{\rho^{(N)}}$ and the corresponding conditional R\'enyi entropies $S_n(R|W)_{\rho^{(N)}}$ may nonetheless depend on the mechanism that generates the black holes. Since this mechanism can in principle be different for different values~$N$, we need to impose a compatibility requirement. Intuitively, it is reasonable to demand that adding another black hole to a collection of $N$ black holes has no noticeable impact on the state of any single black hole. This requirement is analogous to the non-signalling condition~\eqref{eq:pinonsignalling} for path integrals. For our derivations, we will however only use the (weaker) assumption that the radiation $R$ and its relation to $W$ are approximately independent of $N$ for large~$N$. More precisely, we require that the probability distribution $p^{(N)}$ of the states $\rho_{R|w}$ of $R$ conditioned on $W$, which is defined by $\smash{\rho^{(N)}_{R W}}$, is convergent (in distribution) in the limit $N \to \infty$. Since the von Neumann entropy is continuous as a function of the state, this assumption implies that $S(R|W)_{\rho^{(N)}}$ converges in the $N \to \infty$ limit. We may thus define the entropy of a single radiation field conditioned on the reference $W$ in the limit where many black holes are generated by 
\begin{align} \label{eq:condW}
  S(R|W) := \lim_{N \to \infty} S(R|W)_{\rho^{(N)}} \ .    
\end{align}
We note that this definition is robust in the sense that it does not depend on the details of how the quantum de Finetti theorem is employed. In particular, as asserted by Lemma~\ref{lem:uniqueness} in Appendix~\ref{app:conditionalentropyuniqueness}, we could as well have applied the de Finetti theorem to a subset of the $N$ radiation systems rather than to all of them. 


\subsection{Equivalence of swap entropy and conditional entropy} \label{sec:swapcond}

We are now ready to state our first main result, which will help us answering Question~\ref{qn:1}. To establish the link to the phrasing of the question, we first note that Hawking's original argument~\cite{hawking1975particle,hawking1976breakdown} involves an explicit calculation of the particle number distribution in the different modes of the radiation field $R$ of a black hole, from which one may then read off the von Neumann entropy~$S(R)$. Conversely, the quantity that was calculated based on the replica trick is the swap entropy $S^{\mathrm{swap}}(R)_{\mathcal{B}}$ (see Section~\ref{sec:replica} for details). 

\begin{cl}\label{claim:main1}
  Let $\mathcal{B}$ be boundary conditions for a black hole with radiation field~$R$. Then
    \begin{align}\label{eq:main1}
     S^{\mathrm{swap}}(R)_{\mathcal{B}} 
     = S(R|W)
  \end{align}
  where the right hand side is the entropy of $R$ conditioned on the reference~$W$, defined by~\eqref{eq:condW} for an $N$-black-hole system with boundary conditions $\mathcal{B}^{\times N}$.
  \end{cl} 

We remark that, although the quantity~$S(R|W)$ is defined in the $N \to \infty$ limit, Claim~\ref{claim:main1} may be understood as a statement about experiments that require the creation of only a finite number $N$ of black holes. Concretely, for any $\eps > 0$ there exists $N_0$ such that $|{S^{\mathrm{swap}}(R)_{\mathcal{B}}-S(R|W)_{\rho^{(N)}}}|<\eps$ holds for all $N \geq N_0$.

\begin{proof}
By~\eqref{eq:replica_trick_1} and the non-signalling property described in Section~\ref{sec:nonsignalling}, the swap entropy of order~$n$, for any integer, $2 \leq n \leq N$, is given by
\begin{align}
  S^{\mathrm{swap}}_{n}(R)_{\mathcal{B}} 
  = \lim_{N \to \infty} \frac{1}{1-n} \log \langle \tau_n \otimes I^{N-n} \rangle_{\mathcal{B}^{\times N}}\ ,
\end{align}
where $\tau_n$ denotes the cyclic shift operator on $R_1 \cdots R_N$ described in Section~\ref{sec:replica}. Let $\rho_{R_1 \cdots R_N}$ be the joint state of the  $N$ radiation fields corresponding to the boundary conditions $\mathcal{B}^{\times N}$ as defined in Section~\ref{sec:manyblackholes}. Let ${\rho}^{(N)}_{R_1 \cdots R_{N/2} W}$ be the de Finetti state of the form~\eqref{eq:qcstate} that approximates $\rho_{R_1 \cdots R_N}$ as described in  Section~\ref{sec:reference}. We may then express the swap entropy in terms of this state as
\begin{align} \label{eq:swapapprox}
    S^{\mathrm{swap}}_{n}(R)_{\mathcal{B}} 
    =  \lim_{N \to \infty} \frac{1}{1-n}\log  \tr\bigl(\tau_n \sum_w\, p_w^{(N)} \rho^{\otimes n}_{R|w}\bigr)  \ ,
\end{align}
Furthermore, by linearity, the term in the logarithm can be written as
\begin{align}\label{eq:Renyiterm}
      \tr\bigl(\tau_n \sum_w\, p_w^{(N)} \rho_{R|w}^{\otimes n}\bigr)
    =  \sum_w\, p_w^{(N)} \tr(\tau_n \rho_{R|w}^{\otimes n}) 
    =  \sum_w\, p_w^{(N)} \tr(\rho_{R|w}^{n}) 
    = 2^{(1-n)S_n(R|W)_{{\rho}^{(N)}}}\ ,
\end{align}
where we used~\eqref{eq:qcconditionalRenyi} to establish that the sum over~$w$ corresponds to the conditional R\'enyi entropy of order~$n$ evaluated for the state ${\rho}_{R W}^{(N)} = \sum_w p_w^{(N)} \rho_{R|w} \otimes \ketbra{w}{w}_W$. Combining this with~\eqref{eq:swapapprox}, we find that for any integer $n \geq 2$,
\begin{align} \label{eq:swapeq}
    S^{\mathrm{swap}}_{n}(R)_{\mathcal{B}} 
    =  \lim_{N \to \infty} \frac{1}{1-n}\log \Bigl( 2^{(1-n)S_n(R|W)_{{\rho}^{(N)}}} \Bigr) \ .
\end{align}
To proceed, it is convenient to define
\begin{align}
  f(n) & := 2^{(1-n) S_n^{\mathrm{swap}}(R)_{\mathcal{B}}} \\ g_N(n) & := 2^{(1-n) S_n(R|W)_{\rho^{(N)}}} =  \sum_w\, p_w^{(N)} \tr(\rho_{R|w}^n) 
\end{align}
for any $n$ on the complex half plane $\Re(n) \geq 1$. The function $f(n)$ is thus, by construction, given by the analytic continuation of the swap entropy.\footnote{More precisely, the function $f(n) $ is defined as the unique continuation of $2^{(1-n)S_n^{\mathrm{swap}}(R)_{\mathcal{B}}}$ that satisfies the assumptions of Carlson's theorem; see Footnote~\ref{ftn:Carlson} below.} \eqref{eq:swapeq}~tells us that, for integers $n \geq 2$,  the limit $g(n) := \lim_{N \to \infty} g_N(n)$ exists and satisfies 
\begin{align} \label{eq:fgequal}
    f(n) = g(n) \ .
\end{align}
To extend this equality to non-integer values $n>1$, we first note that any function $g_N$ is a finite sum of functions that are analytic in the complex half plane $\Re(n) \geq 1$ and hence itself analytic in this half plane. Furthermore, due to our assumption that the distribution of the conditional states $\rho_{R|w}$ of a single radiation field $R$ is convergent in the limit $N \to \infty$ (see Section~\ref{sec:reference}), Lemma~\ref{lem:distributionconvergence} in Appendix~\ref{app:Renyiconvergence} implies the uniform convergence of the sequence $\{g_N\}_{N \in \mathbb{N}}$ in any compact region within the half plane $\Re(n) \geq 1$.  It then follows from standard theorems of complex analysis that $g = \lim_{N \to \infty} g_N$ is continuous in any such region and hence in the entire half plane $\Re(n) \geq 1$. Similarly, $g$ is analytic in the half plane $\Re(n) > 1$.  Furthermore, $|g|$ is bounded by $1$ in this half plane and hence satisfies the assumptions of Carlson's theorem (see e.g.~\cite{boas2011entire}).\footnote{\label{ftn:Carlson}Carlson's theorem applies to complex functions $f$ that are analytic on the strip $\Re(n) > 1$, continuous on $\Re(n) \geq 1$, and do not grow exponentially fast, i.e., $|f(n)|\le Ce^{\tau|n|}$, for some real constants $C$ and $\tau$, and such that this bound holds with $\tau<\pi$ whenever $\Re(n)=1$. The theorem asserts that such a function is uniquely defined by its values for all but finitely many  $n \in \mathbb{N}$.} We may thus conclude  that~\eqref{eq:fgequal} holds whenever $\Re(n) \geq 1$. But this implies that, for any $n > 1$, 
\begin{align} \label{eq:swapRenyi}
    S_n^{\mathrm{swap}}(R)_{\mathcal{B}} 
    = \lim_{N \to \infty}  S_n(R|W)_{\rho^{(N)} } \ .
\end{align}

To finish our proof we consider the $n \to 1$ limit of~\eqref{eq:swapRenyi}. By the monotonicity of the R\'enyi entropy in $n$ and Lemma~8 in~\cite{tomamichel2009fully}, we have for small enough $n>1$
\begin{align}
    S(R|W)_{{\rho}^{(N)}}\ge S_n(R|W)_{{\rho}^{(N)}}\ge S(R|W)_{{\rho}^{(N)}} - \mathrm{const} \, (n-1) (\log|R|)^2 \ .
\end{align}
Consequently, as $n$ approaches~$1$, $S_n(R|W)_{\rho^{(N)}}$ converges uniformly in $N$ to $S(R|W)_{\rho^{(N)}}$. This means that
\begin{align}
    \lim_{n \to 1} \lim_{N \to \infty} S_n(R|W)_{\rho^{(N)} }
    = \lim_{N \to \infty} \lim_{n \to 1} S_n(R|W)_{\rho^{(N)} }   
    = \lim_{N \to \infty} S(R|W)_{\rho^{(N)}} 
    = S(R|W) \ .
\end{align}
We have thus established that the $n \to 1$ limit of~\eqref{eq:swapRenyi} yields~\eqref{eq:main1}. 
\end{proof}

Claim~\ref{claim:main1} immediately answers the first part of  Question \ref{qn:1}. The replica trick is in fact computing the von Neumann entropy $S(R|W)$ of the radiation $R$ conditioned on the reference~$W$. This quantity is obviously different from the entropy that Hawking calculated, $S(R)$, which does \emph{not} involve any conditioning. The discrepancy between the results of the two calculations is thus due to the fact that they refer to different physical quantities.

Since $S^{\mathrm{swap}}(R)$ follows the Page curve, as shown by the  recent results based on the replica trick~\cite{penington2019replica,almheiri2020replica}, it is clear from Claim~\ref{claim:main1} that the same must be true for the conditional entropy $S(R|W)$; see Fig.~\ref{fig:page1}. In particular, this entropy approaches zero as the black hole evaporates completely. This, in turn, means that the mutual information 
\begin{align} \label{eq:mutual}
  I(R:W) = S(R) - S(R|W)
\end{align}
becomes large for an old black hole. The reference $W$ thus plays a crucial role when we want to understand the dynamics of a black hole.

This conclusion is compatible with the general idea that the discrepancy between Hawking's result and the recent calculations using the replica trick may be due to the fact that the semi-classical path integral captures not one single quantum theory but rather an ensemble of theories~\cite{bousso2020gravity,bousso2020unitarity,engelhardt2021free,pollack2020eigenstate,liu2021entanglement}. To see this, note that, by virtue of~\eqref{eq:qcconditionalvN}, the conditional entropy of the radiation may be expressed as an expectation value, so that 
\begin{align} \label{eq:entropyaverage}
   S^{\mathrm{swap}}(R) = \sum_w p_w S(R|W=w) \ .
\end{align}
Hence, assuming that the black hole behaviour is governed by an ensemble of theories parametrised by $w$, the swap entropy can be interpreted as the ensemble average of the entropy of the radiation state $\rho_{R|w}$ resulting from the evolution prescribed by any particular theory. This picture, which suggests a duality between gravity and an ensemble of quantum theories, is supported by two-dimensional toy models, such as 2D dilaton gravity which is dual to the SYK model defined with a disorder average~\cite{saad2018semiclassical,saad2019jt,saad2019late,stanford2019jt}. Conversely, Hawking's ever-increasing entropy $S(R)$ may  be understood as the entropy of the  mixed state, $\rho_R = \sum_w p_w \rho_{R|w}$,  defined by the effective evolution induced by the ensemble of theories. 

If one has a (possibly effective) theory that allows one to compute the joint radiation state $\rho_{R_1 \cdots R_n}$ of many black holes, one may, following our derivation above, extract expressions for the probabilities $p_w$ and the conditional states $\rho_{R|w}$ from the de Finetti state that approximates $\rho_{R_1 \cdots R_n}$.  A nice example for how $W$ may look like is the black hole toy model of Penington, Shenker, Stanford, and Yang (PSSY)~\cite{penington2019replica}, which has an explicit matrix ensemble as its dual (cf.\ equation~(D.10) in~\cite{penington2019replica}).\footnote{See Section~\ref{sec:pssy} for a short description of the PSSY model.}

The fact that the swap entropy corresponds to an average of entropies  as in~\eqref{eq:entropyaverage} raises the question whether, by typicality, we may also obtain a statement for the individual entropies $S(R|W=w)$.  It is often argued that the entropy calcualted by the replica trick is \emph{self-averaging}~\cite{penington2019replica}, but we are still lacking a general typicality statement. Nonetheless, a criterion for typicality may be obtained from the conditional R\'enyi entropies. Note that, assuming an upper bound of $k$ on the Hilbert space dimension,  the spectrum of any conditional state $\rho_{R|w}$ can be inferred from the first $k-1$ moments which are given by the integer R\'enyi entropies, $S_2(R|{W=w}),\cdots, S_k(R|{W=w})$, using the Newton-Girard method~\cite{song2012,troyer2017}. Furthermore, $2^{(1-n) S_n(R|W)}$, for $n>1$, corresponds to the expectation value of the higher moments of these entropies. One may thus test whether the entropies $S_2(R|W),\cdots, S_k(R|W)$ accurately predict the value of other conditional R\'enyi entropies. This test would succeed if the values $S_n(R|{W=w})$ do not fluctuate depending on~$w$, but may fail otherwise. Such a calculation may be carried out in concrete black hole models, such as the PSSY model, where the R\'enyi entropies can be evaluated by summing over semi-classical saddles.  

Let us now turn to the second part of Question~\ref{qn:1}. The de Finetti approximation $\sum_w  p_w \rho_{R|w}^{\otimes n}$ describes $n$ radiation systems which, individually, are in the same unknown state $\rho_{R|w}$, whereas the probability distribution $p$ serves as a Bayesian prior that captures the uncertainty one may have about that state~\cite{caves2002unknown}. This means that an estimate for the state $\rho_{R|w}$ can be found experimentally by applying state tomography to the $n$ systems (e.g., the scheme proposed in~\cite{christandl2012reliable}, which relies on the sole assumption that the joint state is permutation-invariant). Based on this estimate one may then calculate an approximation for $S(R|{W=w})$.\footnote{There also exist methods to determine $S(R|{W=w})$ more directly, e.g., the Empirical Young Diagram algorithm~\cite{keyl2001estimating,acharya2019measuring}.} Under suitable typicality assumptions, this entropy is also a good estimate for its average over~$w$. 

We conclude that $S(R|W)$ can be determined by tomography in a many-black-hole experiment and thus has a well-defined operational meaning. Conversely, to determine the entropy $S(R)$ from measurements of the Hawking radiation, we would need a state of the form $\left(\sum_w\, p_w \rho_{R|w}\right)^{\otimes n}$. But this is \emph{not} the de Finetti state we would obtain when running an experiment that generates many black holes. We may thus summarise our answer to the last part of Question~\ref{qn:1} as follows: The entropy calculated by Hawking, $S(R)$, cannot be determined experimentally, whereas the entropy obtained via the replica trick, $S^{\mathrm{swap}}(R) = S(R|W)$, can.

\subsection{Equivalence of conditional entropy and regularised entropy} \label{sec:condcorr}

To answer Question~\ref{qn:2}, we prove the following result, which is again a consequence of the quantum de Finetti theorem.

\begin{cl}\label{claim:main2}
  The entropy of the radiation field $R$ of any individual black hole conditioned on the reference $W$, defined by~\eqref{eq:condW} within an $N$-black-hole scenario, satisfies 
  \begin{align} \label{eq:main2}
      S(R|W) = \lim_{N\rightarrow \infty}\frac1N S(R_1\cdots R_N)_{\rho} \ ,
  \end{align}
 where $S(R_1\cdots R_N)_\rho$ is the entropy of the joint state of the radiation fields of all $N$ black holes. 
\end{cl}

Note that, analogously to what is said in the remark after Claim~\ref{claim:main1}, this claim may as well be understood as a statement that holds (approximately) for large but finite~$N$.

\begin{proof}

We start with the direction
 \begin{align} \label{eq:directionone}
      S(R|W) \ge \lim_{N\rightarrow \infty}\frac1N S(R_1\cdots R_N)_{\rho} \ .
  \end{align}
Assume for simplicity that $N = n^2$ for $n \in \mathbb{N}$ and consider a partition of the $N$ radiation systems into $N/n$ blocks of size $n$. Applying the chain rule for the conditional von Neumann entropy recursively, we can decompose $S(R_1\cdots R_N)_{\rho}$ into a sum of entropies for each block,
\begin{multline} \label{eq:entropydecomposition}
    S(R_1\cdots R_N)_{\rho} \\ = S(R_1\cdots R_n)_{\rho}+S(R_{n+1}\cdots R_{2n}|R_1\cdots R_n)_{\rho}+\cdots+S(R_{N-n+1}\cdots R_{N}|R_1\cdots R_{N-n})_{\rho} \ ,
\end{multline}
where the blocks are ordered and the entropy of each block is conditioned on all the previous blocks in the order. The data processing inequality then implies that
\begin{equation}\label{eq:reduction1}
\begin{aligned}
    S(R_1\cdots R_N)_{\rho} \le& S(R_1\cdots R_n)_{\rho}+S(R_{n+1}\cdots R_{2n})_{\rho}+\cdots+S(R_{N-n+1}\cdots R_{N})_{\rho} \\
    =& \frac{N}{n}S(R_1\cdots R_n)_{\rho} \ ,
\end{aligned}
\end{equation}
where the equality holds due to the permutation symmetry of the state $\rho_{R_1\cdots R_N}$.

 Let $\rho_{R_1\cdots R_n W}^{(N)}$ be a de Finetti approximation of the form~\eqref{eq:qcstate}, whose existence is guaranteed by Theorem~\ref{thm:definetti_state}. For each state $\rho_{R|w}$ that appears in the sum that defines this state, let $\psi_{R R'|w}$ be a purification with purifying system $R'$. We may then consider the extension of $\rho_{R_1\cdots R_n W}^{(N)}$ defined by
\begin{equation}
    \rho^{(N)}_{R_1 R_1'\cdots R_nR_n' W} 
    = \sum_w\, \psi_{S|w}^{\otimes n} \otimes \ketbra{w}{w}_W \ .
\end{equation}
Note that the reduced state on $R_1 R_1' \cdots R_n R_n'$ is a convex combination of products of pure states, $\psi_{S|w}^{\otimes n}$. Since these are obviously invariant under permutations, we can conclude that $\rho_{R_1 R_1' \cdots R_n R_n'}$ is supported on the symmetric subspace\footnote{The symmetric subspace of $n$ systems is spanned by state vectors that are invariant under arbitrary permutations.} $\mathrm{Sym}^n(\h_{S})$, which has dimension of order $\O(\mathrm{poly}(n))$, so that its entropy is bounded by $\O(\log n)$. Together with the data processing inequality, this implies that 
\begin{equation}
    I(R_1\cdots R_n :W)_{\rho^{(N)}}
    \leq
    I(R_1R_1'\cdots R_nR_n' :W)_{\rho^{(N)}}
    \leq
    S(R_1R_1'\cdots R_nR_n')
    \leq \O(\log n) \ .
\end{equation}
 The Fannes-Audenaert continuity bound~\cite{fannes1973continuity,audenaert2007sharp} allows us to upper bound the entropy of $\rho_{R_1 \cdots R_n}$ by the entropy of the approximating de Finetti state $\rho^{(N)}_{R_1 \cdots R_n}$, which has distance $\eps=\O(\sqrt{n/N}) = \O(N^{-1/4})$, 
\begin{align}
        S(R_1\cdots R_n)_{\rho} &\le S(R_1\cdots R_n)_{\rho^{(N)}}+\delta \ ,
\end{align}
where $\delta:=\varepsilon n \log|R|+h(\varepsilon)$ and  $h(p):=-p\log p- (1-p)\log (1-p)$ denotes the binary entropy function. Using the definition of the mutual information, we may upper bound the entropy on the right hand side by
\begin{multline}
 S(R_1\cdots R_n)_{\rho^{(N)}}
        = S(R_1\cdots R_n|W)_{\rho^{(N)}}+I(R_1\cdots R_n:W)_{\rho^{(N)}}\\
         \leq S(R_1\cdots R_n|W)_{\rho^{(N)}}+\O(\log n)
        = n S(R|W)_{\rho^{(N)}}+\O(\log n) \ ,
\end{multline}
where we used~\eqref{eq:qcconditionalvN} and the product structure of the de Finetti state conditioned on ${W=w}$ for the last equality. Inserting these inequalities into the decomposition~\eqref{eq:reduction1} and recalling that $n = N^\frac12$, we get 
\begin{equation}
     S(R_1\cdots R_N)_{\rho} \le N S(R|W)_{\rho^{(N)}}+N^\frac12\O(\log N)+N\eps\log|R|+N^\frac12 h(\eps) \ .
\end{equation}
Dividing both sides by $N$, recalling that $\varepsilon = \O(N^{-1/4})$, and taking the limit $N\to \infty$, we arrive at the desired bound~\eqref{eq:directionone}.

It remains to prove the other direction.  We may again decompose the entropy of all $N$ radiation systems into a sum of entropies of blocks as  in~\eqref{eq:entropydecomposition}, where in this case we choose the block size $n=1$. We again assume for simplicity that $N$ is the square of an integer. Consider the $k$th term of this sum, i.e., $S(R_{k}|R_1\cdots R_{k-1})_{\rho}$. If $k \leq K:=N-\sqrt{N}$ we may bound this term using the general version of Theorem~\ref{thm:definetti_state} that includes the auxiliary system $E$. Taking $E$ to contain all systems that can appear in the conditioning, i.e.,  $E=R_1\cdots R_{k-1}$, we still have permutation invariance on $M$ of the other radiation systems, for any $M \leq N-(K-1) = \sqrt{N} +1$. Theorem~\ref{thm:definetti_state} then implies that the state of any $k$ of the radiation systems is approximated, up to an error $\varepsilon' = \O(M^{-1/2})$, by a state $\rho^{(N,M)}_{R_1 \cdots R_{k}}$, whose extension has the form 
\begin{align}\label{eq:generaldFstate}
    \rho^{(N, M)}_{R_{1}\cdots R_{k}W} = \sum_w p^{(N, M)}_w \otimes \rho_{R_1\cdots R_{k-1}|w} \otimes \rho_{R_k|w} \otimes\ketbra{w}{w}_W \ .
\end{align}
The Fannes-Audenaert bound implies
\begin{multline} \label{eq:Slbound}
    S(R_{k}|R_1\cdots R_{k-1})_{\rho} \ge S(R_{k}|R_1\cdots R_{k-1})_{\rho^{(N,M)}}-\delta' \\
    \ge S(R_{k}|R_1\cdots R_{k-1}W)_{\rho^{(N,M)}}-\delta' 
    = S(R|W)_{\rho^{(N,M)}}-\delta' \ ,
\end{multline}
where $\delta':=\eps' \log|R|+h(\eps') = \O(M^{-1/4})$, the second inequality follows from the data processing inequality, and the equality is a consequence of the product structure in~\eqref{eq:generaldFstate}. We now apply the bound~\eqref{eq:Slbound} to the first  $K$ terms on the right hand side of~\eqref{eq:entropydecomposition} and note that the remaining $N-K = \sqrt{N}$ terms are lower bounded by a negative constant to obtain 
\begin{multline}
    S(R_1\cdots R_N)_{\rho}\ge K S(R|W)_{\rho^{(N,M)}}-K \delta' - \O(N^{1/2}) \\
    \geq (N-\sqrt{N}) S(R|W)_{\rho^{(N,M)}} - \O(N M^{-1/4} + N^{1/2}) \ . 
\end{multline}
Finally, we divide by $N$ and take the limit $N\to \infty$,
\begin{align}
    \lim_{N\to \infty}\frac1N S(R_1\cdots R_N)_{\rho}
    \geq  \lim_{N \to \infty} S(R|W)_{\rho^{(N,M)}} + O(M^{-1/4}) \ .
\end{align}
The desired bound then follows from Lemma~\ref{lem:uniqueness}.
\end{proof}

The right hand side of~\eqref{eq:main2} may be understood as a \emph{regularised} entropy. It measures the entropy that any individual radiation system~$R_i$ contributes to the joint entropy of all radiation systems in the many-black-hole scenario. Crucially, the reference~$W$, which appears as the conditioning system in the von Neumann entropy on the left hand side of~\eqref{eq:main2}, has no explicit occurrence in this regularised entropy. The reference-dependence is thus somehow hidden in the joint state, and hence the correlation, of the radiation fields $R_1, \ldots, R_N$ of the independently prepared black holes. In this sense, the reference can be regarded as a  global property of spacetime.

Claim~\ref{claim:main2} also sheds light on the role of the replica black holes, which occur in the calculation of~$S^{\mathrm{swap}}(R)_{\mathcal{B}}$. For this we combine it with Claim~\ref{claim:main1} to obtain
\begin{equation}\label{eq:main3}
    S^{\mathrm{swap}}(R)_{\mathcal{B}} = \lim_{N\rightarrow \infty}\frac1N S(R_1\cdots R_N)_{\rho}.
\end{equation}
The replica trick thus calculates the regularised entropy of the black hole radiation. We can conclude from this that the introduction of replicas is not just a purely mathematical trick to compute the von Neumann entropy $S(R)$ --- the trick simply doesn't yield this quantity! The replicas should instead be understood as physically relevant objects, for they are needed to give a meaning to the right hand side of~\eqref{eq:main3}.

Finally, let us turn to the second part of Question~\ref{qn:2}, which asks how the gravitational path integral, within a semi-classical approximation, can ``know'' about the reference~$W$ to yield $S(R|W)$ and thus reproduce the Page curve.  For this we note that any $W$ that is non-trivial, in the sense that $S(R)_\rho \neq S(R|W)$, implies that the individual black hole systems are correlated. More precisely, for $N$ large,
\begin{multline}
  S(R)_{\rho} - S(R|W)_{\rho} \approx S(R)_\rho - \frac{1}{N} S(R_1 \cdots R_N)_{\rho} 
  = \frac{1}{N} \sum_{n=1}^N \bigl( S(R_n)_{\rho} - S(R_n|R_1 \cdots R_{n-1})_{\rho} \bigr) \\
  = \frac{1}{N} \sum_{n=1}^N I(R_n: R_1 \cdots R_{n-1})_{\rho}
  = \frac{1}{N} \sum_{n=1}^N I(R_1: R_2 \cdots R_{n})_{\rho}
\end{multline}
where we have used the chain rule for the von Neumann entropy and where the last equality follows from permutation-symmetry. The right hand side is an average over~$n$ of the mutual information $I(R_1 : R_2 \cdots R_n)$. Since this mutual information is monotonically non-decreasing in $n$ and bounded by $2S(R_1)$, the average over $n$ may, for large~$N$, be replaced by any typical~$n \in \{1, \ldots, N\}$.  Hence, for such~$n$, the mutual information between one radiation field $R_1$ and $n-1$ other radiation fields satisfies
\begin{align} \label{eq:systemcorrelation}
    I(R_1: R_2 \cdots R_n)_\rho \approx I(R_1 : W)_\rho = S(R)_\rho - S^\mathrm{swap}(R) \ . 
\end{align}
This confirms that the correlation between the individual radiation systems becomes non-zero after the Page time, when the two entropies on the right hand side start to diverge (Fig.~\ref{fig:page1}). But this is precisely when the wormhole solutions become dominant. Following the spirit of the ER=EPR conjecture~\cite{van2010building,maldacena2013cool}, one might more generally expect that correlations (rather than only entanglement) manifest themselves as particular spacetime geometries. We may therefore regard the wormhole solutions as  geometrical representations of the correlation between the radiation systems, which is mediated by the reference~$W$. This suggests, as an answer to Question~\ref{qn:2}, that it is the replica wormhole geometry that ``knows'' about~$W$. 

\section{From the typical unitary to the Elusive Reference (ER) model}\label{sec:model}

An important landmark in understanding the information-theoretic content of Hawking radiation is Page's derivation of the Page curve~\cite{page1993average,page1993information}. He championed a model where the time evolution of a black hole as viewed by an asymptotic observer corresponds to a typical unitary, i.e., a unitary that one would pick almost surely according to the Haar measure over the set of all possible unitaries. A further important step towards an information-theoretic understanding of Hawking radiation was the work by Hayden  and Preskill (HP). Using basically the same model as Page's, they showed that a black hole acts like a mirror, bouncing back any information thrown into it after the Page time~\cite{hayden2007black}. The typical unitary model highlights the idea that a black hole appears from the outside as a maximally chaotic system that scrambles information that falls into it~\cite{sekino2008fast,shenker2014black,maldacena2016bound,cotler2017black,stanford2019jt,hayden2016holographic,belyansky2020minimal,krishnan2021hints,swingle2016measuring,landsman2019verified,pollack2020eigenstate,liu2021entanglement}. The model and its variants also inspired experiments for witnessing quantum gravity effects in the lab~\cite{swingle2016measuring,landsman2019verified,brown2019quantum,nezami2021quantum}. 

The success of the information-theoretic approach pioneered by Page and HP raises the question whether the results presented in Sections~\ref{sec:swapcond} and~\ref{sec:condcorr} are compatible with the typical unitary model. To answer this question we first introduce an extension of the model to the many-black-hole scenario described in Section~\ref{sec:manyblackholes}.

\subsection{The Elusive Reference (ER) model} \label{sec:elusivereference}

Consider a many-black-hole system and let $W$ be the reference constructed as described in Section~\ref{sec:reference}. We have already established several facts about the nature of $W$. First, since $I(R:W) = 0$ before the Page time, $W$ does not tell us anything about the early radiation. Conversely, the early radiation does not tell us anything about $W$. Second, since $S(R|W) = 0$ when the black hole is evaporated completely, $W$ must contain full information about the final state of any radiation field. Furthermore, it follows from~\eqref{eq:systemcorrelation} that this information can be retrieved from the other radiation fields.

Taken together, these facts suggest the following model of a many-black-hole system, which we will refer to as the \emph{Elusive Reference (ER)} model. Suppose that we start with an initial situation consisting of $N$ identical subsystems $M_1, \ldots, M_N$ containing initial matter shells. Each matter shell $M_i$ collapses to a black hole $B_i$ that emits radiation $R_i$. We may then describe the evolution from the initial time to some fixed later time as an isometry $\$$ from $M_1 \dots M_N$ to $B_1 R_1 \cdots B_N R_N W$ defined by\footnote{The notation $\$$ for this map is deliberately taken to resemble the map defined by Eq.~4.25 of~\cite{bartlett2007}. The notation $\$$ is also used in~\cite{hawking1982unpredictability} to denote the (non-unitary) $S$-matrix that describes the black hole evolution.}
\begin{align} \label{eq:dollarmap}
    \$ :=  \sqrt{\frac{1}{|\mathcal{W}|}} \sum_{w \in \mathcal{W}} \,  (U_w)^{\otimes N} \otimes \ket{w}_W  \ ,
\end{align}
where $\{U_w\}_{w \in \mathcal{W}}$ is a quantum $N$-design of isometries from  $M_i$ to the joint system $B_i \otimes R_i$.\footnote{A finite family of unitaries $\{U_w\}_{w \in \mathcal{W}}$ is a quantum $N$-design if $\smash{\frac{1}{|\mathcal{W}|}\sum_w (U_w\ketbra{\psi}{\psi}U_w^{*})^{\otimes N}} = \smash{\int\dd U\,(U\ketbra{\psi}{\psi}U^{*})^{\otimes N}}$ for any $\ket{\psi}$, where $\dd U$ denotes the Haar measure. Here we use the canonical extension of this definition from the set of unitaries on a system to the set of isometries. In the limit of $N\to\infty$, a sequence of $N$-designs corresponds to the Haar measure.} We will generally think of $W$ as a reference system that cannot be accessed directly, i.e., we do not assume that there exist physical operations to act on it or to measure it.

Suppose that we are only interested in the action of $\$$ on $n \ll N$ systems and assume that the input to the other $N-n$ systems is random.\footnote{It is sufficient that the input is chosen at random within the subspace spanned by matter shell states that one considers for the first $n$ subsystems.} It then follows from the de Finetti theorem, applied to the Choi-Jamio\l{l}kowski representation of $\$$~\cite{jamiolkowski1972linear,choi1975completely}, that the reduced map approximately takes the form (written as a TPCPM)
\begin{align} \label{eq:dollarmapreduced}
    X_{M_1} \otimes \cdots \otimes X_{M_n} \quad \mapsto \quad \frac{1}{|\mathcal{W}|} \sum_{w \in \mathcal{W}} \,  U_w X_{M_1} U_w^{*} \otimes \cdots \otimes U_w X_{M_n} U_w^{*} \otimes \ketbra{w}{w}_W \ .
\end{align}
Note that this map is not unitary and hence not reversible. It corresponds to a ``twirling'' map as considered in the literature on quantum reference frames~\cite{Aharonov1967,kitaev2004,bartlett2007}. There the aim is to include reference frames (e.g., for position or for directions) explicitly into the description of a quantum system. 

The analogy to quantum reference frames is useful to illustrate the ER model and to motivate the choice of the terminology ``elusive''. As a typical example, consider a system consisting of $n$ spins that point in spatial directions $X_1, \ldots, X_n$. The reference frame relative to which these directions are defined is modelled as a separate system $W$, which one may think of as a tripod. In this analogy, the twirling map would have the form~\eqref{eq:dollarmapreduced}, with the sum replaced by an integral over the rotation group. The twirling captures the idea that, without access to any other subsystems, the direction of an individual spin, say $X_1$, would look completely random to us. Crucially, however, any sufficiently large subset of the other spins can serve as a physical reference frame. Hence, relative to them, $X_1$ is a well-defined direction. Note that the latter is true even if we do not have direct access to~$W$, i.e., the tripod may be a purely hypothetical construct. But even if we did think of~$W$ as a physical tripod, we could not tell the direction of this tripod unless we had another (super-)reference frame relative to which it is defined. For these reasons, it is sensible to regard $W$ as an elusive system that we cannot experimentally act on.

In the ER model, it is the Hawking radiation that takes the role of the spin direction in the example above. Note that the corresponding reference $W$ would then not  need to define  directions in real space, but in the Hilbert space of the entire radiation field~$R$. Hence, lacking access to such a reference, the radiation of any single black hole looks completely random to us. Its physical state may even be undefined, in the same sense as the direction of a single spin in a universe without a direction reference would be undefined. Conversely, radiation fields of other black holes can serve as a reference, in the same way as the direction of a spin can be defined relative to a collection of other spins.

Let us conclude this subsection by comparing the ER model to the typical unitary model as considered by Page and HP. If we reduce the ER model to a situation featuring only a single black hole and omit the reference information~$W$, then the reduced map~\eqref{eq:dollarmapreduced} for $n=1$ corresponds to an averaging over typical unitaries. However, the physical interpretation of this map is rather different from that admitted by Page and HP. For example, in their work on black holes as mirrors~\cite{hayden2007black}, HP adopt the view that the black hole dynamics is typical but fixed. According to this view, $W$ could be regarded as a (fixed) parameter of the theory that describes the black hole physics. Conversely, according to the ER model, the isometry $U_w$ cannot be fixed, for this would imply $I(R_i:W) = 0$ even after the Page time, thus contradicting our results from above. This shows that $W$ cannot be regarded as a fixed parameter. Nonetheless, HP's conclusion that a black hole acts as a mirror can be restored approximately in the ER model, but we would necessarily have to consider many black holes. The radiation fields of the extra black holes would then serve as an approximate reference, relative to which the mirrored information can be decoded (see also the discussion following Proposition~\ref{pr:Schmidt} below).

\subsection{Black hole interior} \label{sec:toy}

In the information-theoretic approach by Page and HP, the black hole  interior appears explicitly as a quantum system, $B$. In the ER model it is however not clear that one should regard the $N$ systems $B_1, \ldots, B_N$ as separate interiors of the individual black holes. To enable a distinction between different notions of the interior, we will refer to $B_i$ as the \emph{primary interior} of the $i$th black hole.  Analogously to the HP model, we take the dimension of each $B_i$ to be equal to $2^{S_{\mathrm{BH}}}$, where $S_{\mathrm{BH}}$ is the black hole's (thermal) Bekenstein-Hawking entropy, and hence a measure for its size. 

According to the ER model, the joint state of the $N$-black-hole system obtained from $N$~matter shells in an initially pure state $\ket{0}_M$ is
\begin{align}
    \rho_{B_1R_1 \cdots B_N R_N W} =  \$ \bigl(\ketbra{0}{0}_{M}\bigr)^{\otimes N} \$^{*}\ .
\end{align}
Let us first focus on one single black hole~$B=B_1$ and its radiation field~$R=R_1$. Assuming that $N \gg 1$, the joint state of $B$ and $R$ together with the reference $W$ is approximately a mixed state of the form (see also~\eqref{eq:dollarmapreduced})
\begin{align} \label{eq:BRWstate}
    \rho_{BRW} = \tr_{B_2 R_2 \cdots B_N R_N}(\rho_{B_1R_1 \cdots B_N R_N W}) \approx \frac{1}{|\mathcal{W}|} \sum_{w \in \mathcal{W}} \,  U_w \ketbra{0}{0}_M U_w^{*} \otimes \ketbra{w}{w}_W \ .
\end{align}
Following Page's analysis, one may use typicality (or, more generally, decoupling theorems~\cite{hayden2007black,dupuis2010decoupling}) to show that, for any fixed $w \in \mathcal{W}$, the state on $B$ and $R$ conditioned on $W=w$, which is by construction pure, must approximately (for  $k+S_{\mathrm{BH}} \gg 1$) have a Schmidt decomposition of the form
\begin{align} \label{eq:U0decomposition}
  U_w \ket{0}_M 
  \sim \sum_j\ket{j_w}_B\ket{j_w}_R \ ,
\end{align}
where $\{\ket{j_w}_B\}_{j}$ and $\{\ket{j_w}_R\}_{j}$ are families of orthonormal states on $B$ and $R$, respectively. 

If one has access to a single black hole only, then $W$ is unknown according to the ER model. This is a difference to Page's and HP's model. To account for this difference, we do not fix $W=w$ but instead describe the situation with respect to a fixed orthonormal basis $\{\ket{i}_R\}_{i \in \{1, \ldots, 2^k\}}$, which may be interpreted as a measurement basis for the radiation field~$R$. (We assume for convenience that $R$ has the size of $k$ qubits.) It is useful to consider a purification of  $\rho_{R B W}$, which we define using a second copy, $\bar{W}$, of the reference~$W$, 
\begin{align} \label{eq:BRWpurification}
   \ket{\rho}_{W\bar{W}BR} 
    = \frac{1}{\sqrt{|\mathcal{W}|}}\sum_{w \in \mathcal{W}} \ket{w}_W\ket{w}_{\bar{W}} U_w \ket{0}_M \ .
\end{align}
Due to the mixing property of the $N$-design  $\{U_w\}_{w \in \mathcal{W}}$, the reduced state on $R$ is maximally mixed. Since, the reduced state on $W \bar{W}$ is constrained to the subspace spanned by vectors of the form $\ket{w}_W \ket{w}_{\bar{W}}$ for $w \in \mathcal{W}$, we immediately find the following statement about the Schmidt decomposition of $\ket{\rho}_{W\bar{W}BR}$. 
 
\begin{pr} \label{pr:Schmidt}
For any fixed basis $\{\ket{i}_R\}_{i \in \{1, \ldots, 2^k\}}$, there exists an orthonormal family $\{\ket{\psi_i}_{W \bar{W} B}\}_{i \in \{1, \ldots, 2^k\}}$ of states of the form
\begin{align} \label{eq:BWstates}
  \ket{\psi_i}_{W\bar{W}B}=\frac{1}{\sqrt{|\mathcal{W}|}}\sum_{w \in \mathcal{W}} \ket{w}_W\ket{w}_{\bar{W}}\ket{\psi_i^w}_B \ ,
\end{align}
where $\ket{\psi_i^w}_B$ are (not necessarily normalised) vectors, such that
\begin{align}\label{eq:pssystate}
    \ket{\rho}_{W\bar{W}BR} 
    = \frac{1}{\sqrt{2^k}}\sum^{2^k}_{i=1}\ket{\psi_i}_{W\bar{W}B}\ket{i}_R \ .
\end{align}
\end{pr}


One may interpret the joint system $W \bar{W} B$ as the \emph{effective  interior} of the black hole and regard $\ket{\psi_i}_{W \bar{W} B}$ as the interior partners of the radiation states $\ket{i}$. Note that these states  genuinely exhibit entanglement between the primary interior $B$ and the reference. Hence, in contrast to the situation considered by Page and HP, the external radiation $R$ is not purified by the primary interior $B$ alone. 

To reproduce the conclusions by Page and HP about the state of the primary interior~$B$, we first need to return to the joint state of $n \gg 1$  black holes and their radiation fields. Using the fact that this state approximately has de Finetti form (if $n \ll N$) as well as Proposition~\ref{pr:Schmidt}, we find
\begin{equation} \label{eq:globalstate}
    \rho_{WB_1R_1 \cdots B_nR_n} \approx  \frac{1}{|\mathcal{W}|} \sum_{w \in \mathcal{W}} \ketbra{w}{w}_W \otimes \Bigl(\frac{1}{\sqrt{2^{k}}} \sum_{i=1}^{2^k} \ket{\psi^w_i}_B\ket{i}_R \Bigr)^{\otimes n} \Bigl(\frac{1}{\sqrt{2^{k}}} \sum_{i=1}^{2^k} \bra{\psi^w_i}_B\bra{i}_R \Bigr)^{\otimes n}\, .
\end{equation}
Tracing out the primary interiors of the $n$ black holes, the marginal state on $W R^n$ reads
\begin{align} \label{eq:globalstatereduced}
    \rho_{WR_1 \cdots R_n} 
    = \frac{1}{|\mathcal{W}|} \sum_{w \in \mathcal{W}} \ketbra{w}{w}_W\otimes\rho_{R|w}^{\otimes n} \quad ,
\end{align}
where, for any fixed $w$, the state $\rho_{R|w}$ may be expressed in terms of basis elements $\ket{i}_R$ as
\begin{align} \label{eq:internalstateoverlap}
  \rho_{R|w} =\frac{1}{2^k}\sum_{i,j=1}^{2^k}\braket{\psi^w_i}{\psi^w_j}_B\ketbra{j}{i}_R \ .
\end{align}
The radiation fields of the $n$ black holes are thus each in the same state $\rho_{R|w}$, which may however be unknown because the reference $W$ is not assumed to be accessible. Nonetheless, an estimate for $\rho_{R|w}$ can be obtained by applying quantum state tomography on some of the radiation systems.  Conditioned on this estimate, the overlap 
\begin{align} \label{eq:internaloverlap}
   \braket{\psi^w_i}{\psi^w_j}_B = 2^k \bra{i} \rho_{R|w} \ket{j} 
\end{align}
can be calculated and thus represents an operationally accessible quantity that provides information about the structure of the states of the primary interior~$B$ of any single black hole. 

More concretely, according to~\eqref{eq:globalstate}, each vector $\ket{\psi^w_i}_B$ corresponds (up to normalisation) to the primary interior partner state of the radiation state $\ket{i}_R$, conditioned on the reference $W=w$. While the latter is not physically accessible, equation~\eqref{eq:internaloverlap} tells us that, in order to make any claims about how the different interior states $\ket{\psi^w_1}_B, \ldots, \ket{\psi^w_{2^k}}_B$ relate to each other, it is sufficient to know $\rho_{R|w}$. In particular, conditioned on this knowledge, which as argued above is physically accessible in the many-black-hole scenario via quantum state tomography, the state of~$B$ corresponding to any radiation mode $\ket{i}_R$ is pure. The operationally accessible knowledge in the ER model is thus basically equivalent to knowing~$w$, and hence~$U_w$. Conditioned on this knowledge, a black hole in the ER model can thus be treated in the same way as in Page's or HP's model, where the unitary that describes the time evolution (corresponding to $U_w$ in the ER model) is fixed. 

This implies that the Page curve can be reproduced in the ER model, provided that we condition on the knowledge obtained via tomography as just described. Indeed, according to~\eqref{eq:globalstate}, the joint state of $B$ and $R$ conditioned on~$W=w$ is pure. Consequently, the entropy of the density operator $\rho_{R|w}$ cannot be larger than the size $S_{\mathrm{BH}}$ of~$B$. In particular, after the Page time, i.e., when $k > S_{\mathrm{BH}}$, the overlap~\eqref{eq:internaloverlap} cannot be zero for all $i \neq j$. Hence, in contrast to the states $\ket{\Psi^1}_{W \bar{W} B}, \ldots,  \ket{\psi_{2^k}}_{W \bar{W} B}$ of the effective interior $W \bar{W} B$, the vectors $\ket{\psi^w_1}_B, \ldots, \ket{\psi^w_{2^k}}_B$ on the primary interior~$B$ are not mutually orthogonal. 

We may conclude that the ER model is compatible both with Hawking's calculation and with the Page curve. The entropy $S(R|W)$, which must have the same behavior as $S({R|W=w})$ for a typical choice of ${W=w}$, follows the Page curve. In particular, it decreases after the Page time, and this decrease is  reflected by the non-orthogonality of the family $\{\ket{\psi^w_i}_B\}_{i \in \{1, \ldots,  2^k\}}$ of  states of the primary interior~$B$. Conversely, the orthogonality of the family $\{\ket{\psi_i}_{W \bar{W} B}\}_{i \in \{1, \ldots, 2^k\}}$ of states of the effective interior $W \bar{W} B$ implies that the entropy on $S(R)$ must be maximal, as predicted by Hawking. 

\subsection{Comparison to the PSSY model}\label{sec:pssy}

The evaluation (and even definition) of path integrals such as~\eqref{eq:partitionobservable} is difficult for realistic theories consisting of interacting fields $\psi$ on a $(3+1)$-dimensional spacetime. The purpose of the model proposed by PSSY~\cite{penington2019replica} is to simplify the analysis to a degree where explicit calculations of path integrals become possible, while keeping qualitative features of black hole physics. To this aim, they take inspiration from Jackiw-Teitelboim (JT) gravity, which is a $(1+1)$-dimensional theory of gravity that couples to a \emph{dilaton field} $\phi$.

\begin{figure}
\centering
\begin{subfigure}{.33\textwidth}
  \centering
  \includegraphics[width=0.8\textwidth]{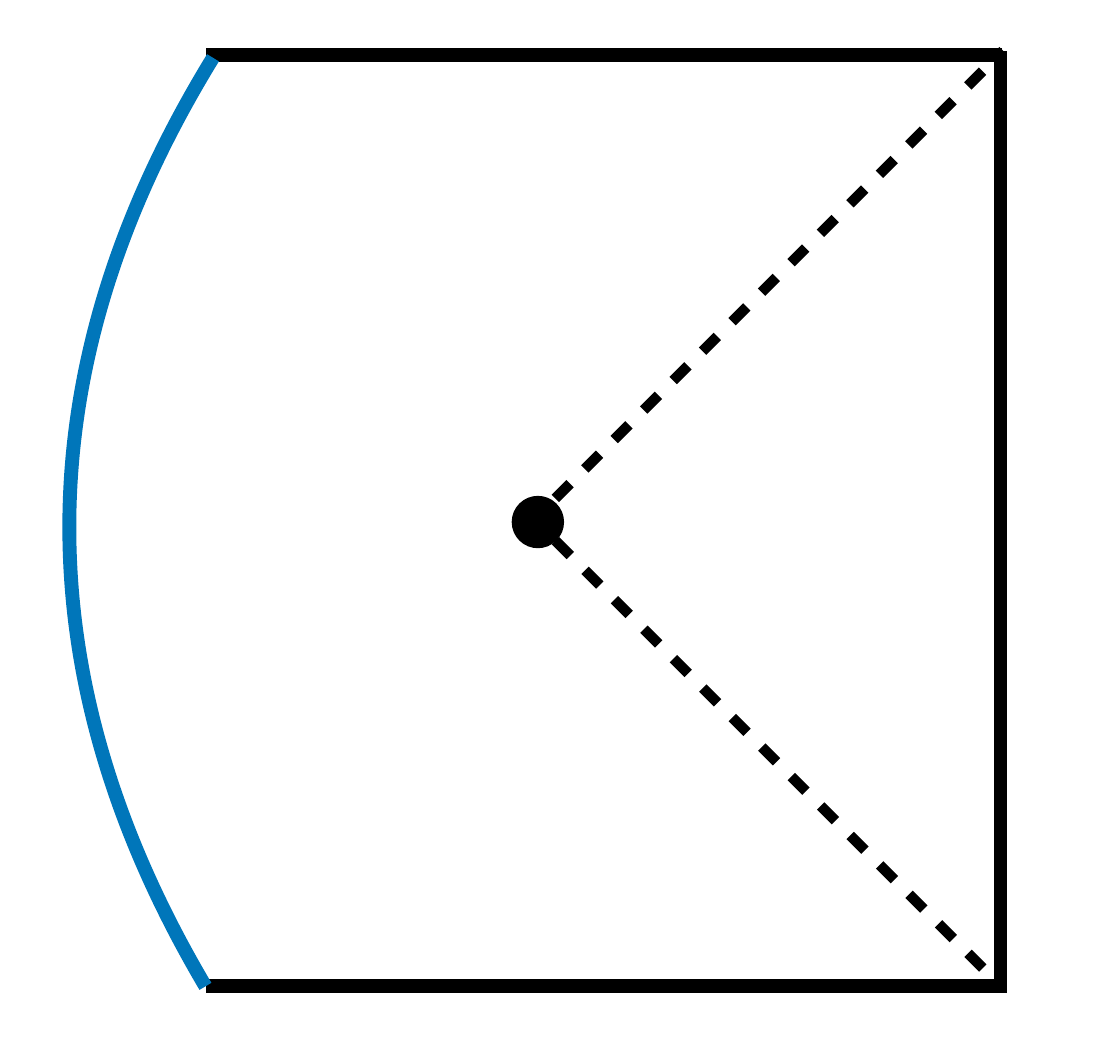}
   \caption{Lorentzian spacetime.}
   \label{fig:lorentz}
\end{subfigure}
\begin{subfigure}{.31\textwidth}
  \centering
  \includegraphics[width=0.8\textwidth]{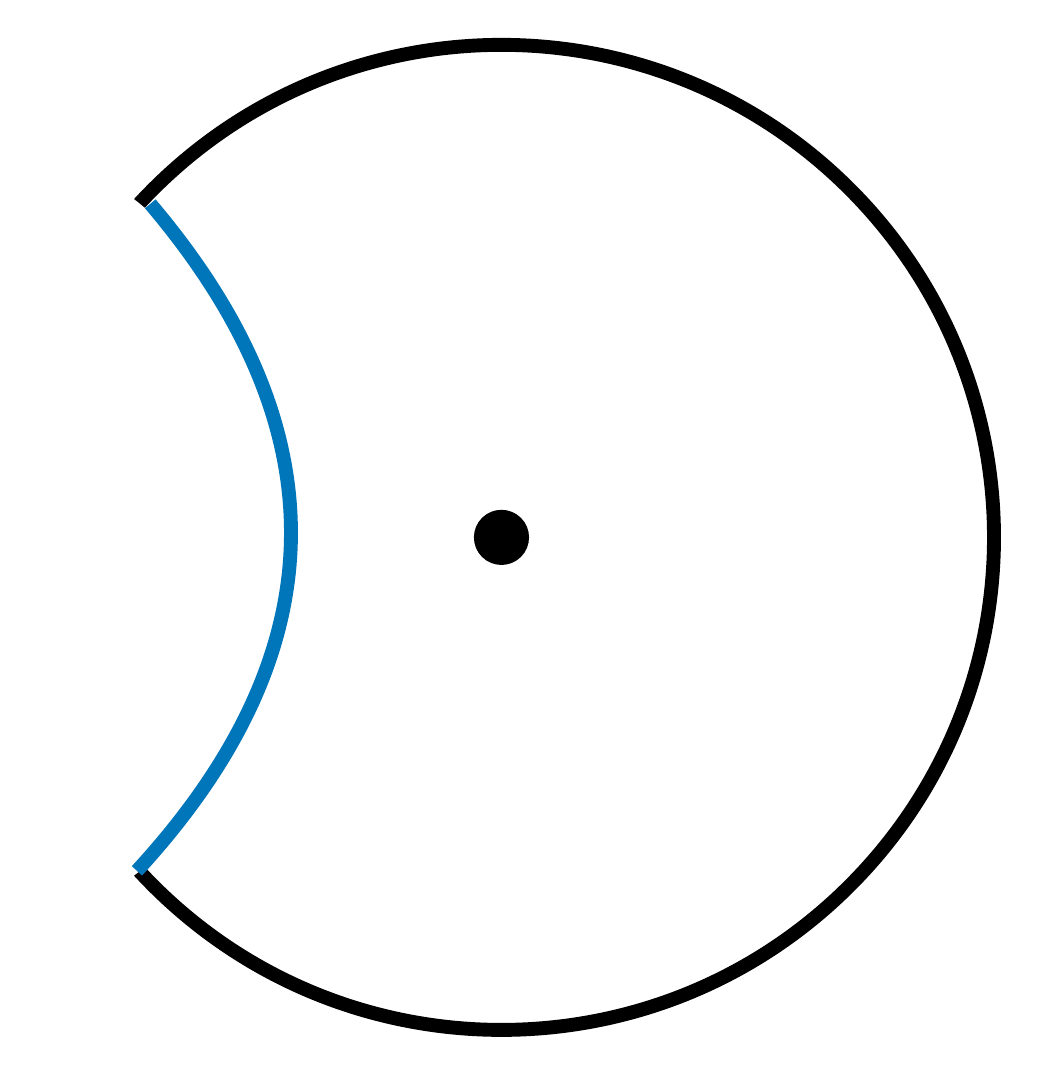}
    \caption{Euclidean spacetime.}
   \label{fig:euclid}
\end{subfigure}%
\caption{Schematic of the spacetime considered within the PSSY model in Lorentzian and Euclidean metric. 
             Fig.~\ref{fig:lorentz} shows the spacetime containing a black hole. The dashed lines represent the bifurcate event horizon, the solid vertical line spatial infinity. The blue line represents the EOW brane. Since path integrals are difficult to define and evaluate for a Lorentzian spacetime, one considers an Euclidean analogue, as shown on Fig.~\ref{fig:euclid}. }
    \label{fig:geometry}
\end{figure}

Concretely, the Euclidean action is taken to be~\cite{harlow2020factorization,penington2019replica,almheiri2020replica} 
\begin{align} \label{eq:JTaction}
    I_{\mathrm{JT}}[g, \phi] = - \frac{S_0}{2\pi} \Bigl( \frac{1}{2} \int_{\mathcal{M}} \sqrt{g} R + \int_{\partial \mathcal{M}} \sqrt{h} K \Bigl)
    - \Bigl( \frac{1}{2} \int_{\mathcal{M}} \sqrt{g} (R+2) \phi + \int_{\partial \mathcal{M}} \sqrt{h} (K-1) \phi  \Bigl) \ ,
\end{align}
where the integrals over $\mathcal{M}$ correspond to the Einstein-Hilbert action\footnote{The gravitational coupling constant $G_N$ is conventionally set to one in 2D gravity coupled to a dilaton field.} for the metric $g$ with the curvature scalar denoted by $R$, and where the integrals over the topological boundary $\partial \mathcal{M}$ of $\mathcal{M}$ correspond to the Gibbons-Hawking-York term with the induced metric $h$ and the trace over the extrinsic curvature $K$. $S_0$ is a parameter that will correspond to the extremal entropy of the black hole solution, i.e., the black hole entropy in the zero temperature limit (see below). 

In its Lorentzian version, which serves as the physical guideline, the spacetime $\mathcal{M}$ has an asymptotic boundary region, which includes the regions outside of the black hole; see  Fig.~\ref{fig:lorentz}. In addition, one bounds $\mathcal{M}$ at the interior by an \emph{end of the world (EOW) brane}~\cite{kourkoulou2017pure}, which intersects the asymptotic boundary in the past and future. The EOW brane may be interpreted as the worldline of a particle of mass $\mu \geq 0$ that moves freely on the ambient spacetime, i.e., its dynamics is governed by the action
\begin{align}
    I_{\mathrm{particle}}[g] = \mu \int_{\mathrm{EOW}} \dd s \ .
\end{align}
where $\dd s$ is the line element of the metric $g$.

The Euclidean spacetime considered by PSSY, which is shown by  Fig.~\ref{fig:euclid}, also has an asymptotic boundary and a boundary defined by the EOW. The boundaries intersect at two points which, in analogy to the Lorentzian picture, one may call \emph{infinite past} and \emph{infinite future}. At the asymptotic boundary one imposes the condition $\phi = \frac{1}{\eps}$, for $\eps \to 0$, and defines the boundary metric by $\dd\tau^2 = \eps^2 ds^2\big|_{\partial \mathcal{M}}$. At the EOW, the boundary conditions are $K=0$ and $\partial_n \phi = \mu$, where $\partial_n$ denotes the derivative normal to the boundary. Finally, the renormalised length $\beta$ of the asymptotic boundary between the intersection points with the EOW is fixed. This set of boundary conditions admits a Euclidean AdS-Schwarzschild black hole solution with temperature $\beta^{-1}$ and Bekenstein-Hawking entropy $S_{\mathrm{BH}}=S_0+\phi_h$, where $\phi_h$ is the dilaton value at the horizon, which is proportional to the temperature. For simplicity, we continue our discussion with a black hole at low temperature, so that $S_{\mathrm{BH}}$ is well approximated by $S_0$.

\begin{figure}
\centering
         \includegraphics[width=0.9\textwidth]{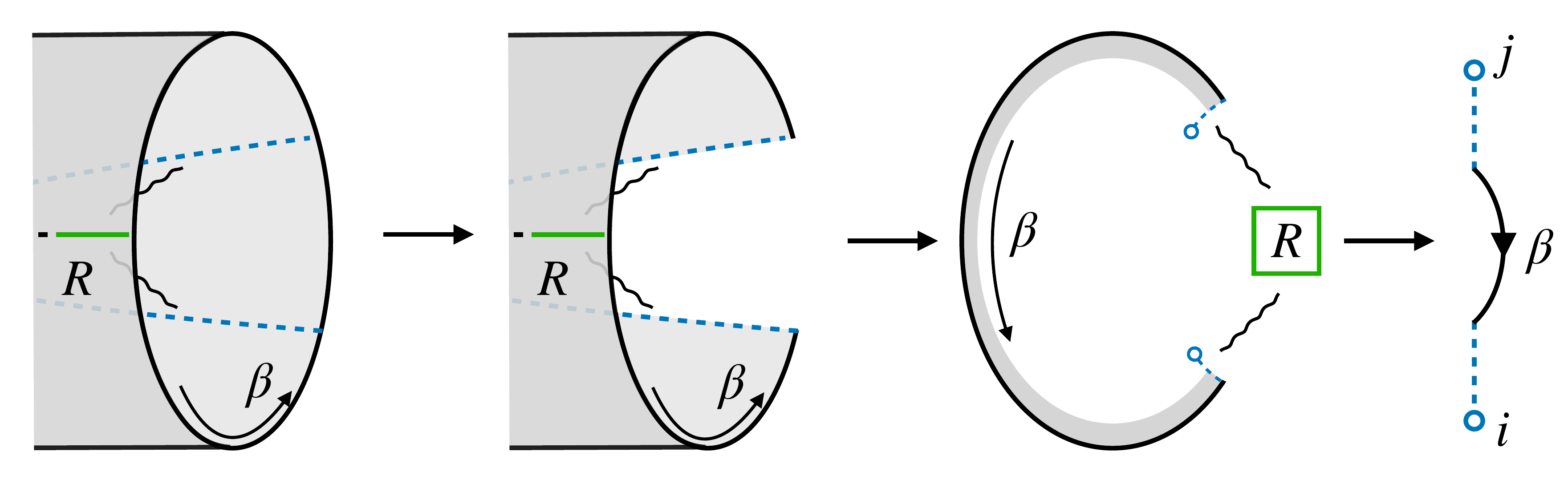}
             \caption{Boundary conditions in the PSSY model. The sequence of figures illustrates how the boundary conditions of the PSSY model (commonly represented by a diagram as shown on the rightmost) may be understood in the context of the general scenario shown in Fig.~\ref{fig:bc} (leftmost here). The blue dashed line represents the worldline (brane) of a particle in the asymptotically hyperbolic space. The curly lines represent the entanglement between the Hawking radiation $R$ and the black hole interior. In a first simplification step (first arrow), the worldline of the particle is made the end of the world (EOW) by deleting the spacetime beyond it and regarding the interior degrees of freedom of the black hole as internal degrees of freedom of the particle.  In a second simplification step (second arrow), the radiation is taken out of the spacetime and treated as a separate reservoir~$R$. The correlation between the radiation and the EOW particle is imposed by matching any content $\ketbra{i}{j}_R$ of $R$ to the internal degrees of freedom $i$ and $j$ of the particle at the infinite past and future, respectively. }
    \label{fig:radiationembedding}
\end{figure}

In the PSSY model, the Hawking radiation is not generated dynamically. Rather, one assumes the existence of a radiation reservoir $R$ that is external to the JT spacetime and consists of $k$ qubits. It can be interpreted as a system in which one has stored all the radiation that has been emitted. To model the interior partners of the Hawking radiation quanta, one assumes that the EOW particle has $2^k$ distinguishable internal states, labelled by  $i \in \{1, \ldots, 2^k\}$. This degree of freedom undergoes no separate dynamics, i.e., the internal state~$i$ remains constant on the particle's worldline. Since the worldline of the EOW particle intersects the asymptotic boundary, the internal state can be set by appropriate boundary conditions. The idea that the EOW particle is the interior partner of the Hawking radiation is now imposed by connecting these boundary conditions to the states of $R$. Specifically, one equips $R$ with an orthonormal basis $\{\ket{i}_R\}_{i \in \{1, \ldots,  2^k\}}$ and identifies $\ket{i}_R$  with the boundary condition that the EOW brane has internal state $i$ in the infinite past, and any $\bra{j}_R$ with the boundary condition that the EOW brane has internal state $j$ in the infinite future.  

This idea yields a rule for evaluating the expectation value of any observable $O = \sum_{i,j} O_{ij} \ketbra{j}{i}_R$ on the radiation reservoir~$R$. Concretely, the partition function (see~\eqref{eq:partitionobservable}) is given by 
\begin{align}
    Z[\mathcal{B}, O] = \sum_{i,j} O_{ij}  \int_{i \stackrel{\beta}{\longrightarrow} j} \mathcal{D}g\mathcal{D}\psi e^{-I_{\mathrm{JT}}[g,\phi] - I_{\mathrm{particle}}[g]}
\end{align}
where the subscript to the integral indicates that the boundary conditions $\mathcal{B}$ demand that the internal state of the interior particle in the asymptotic past and future is $i$ and $j$, respectively, and that the asymptotic spacetime boundary between these two points has renormalised length~$\beta$ (see Fig.~\ref{fig:radiationembedding}). Note that only the end points of the EOW particle are visible at the boundary, but not their internal connection via the EOW worldine. This fact is important when replicas are considered. In this case a connection of two end points via an asymptotic boundary does \emph{not} imply that they belong to the same EOW particle. 

\begin{figure}
\centering
\begin{subfigure}{.33\textwidth}
  \centering
  \includegraphics[width=.95\linewidth]{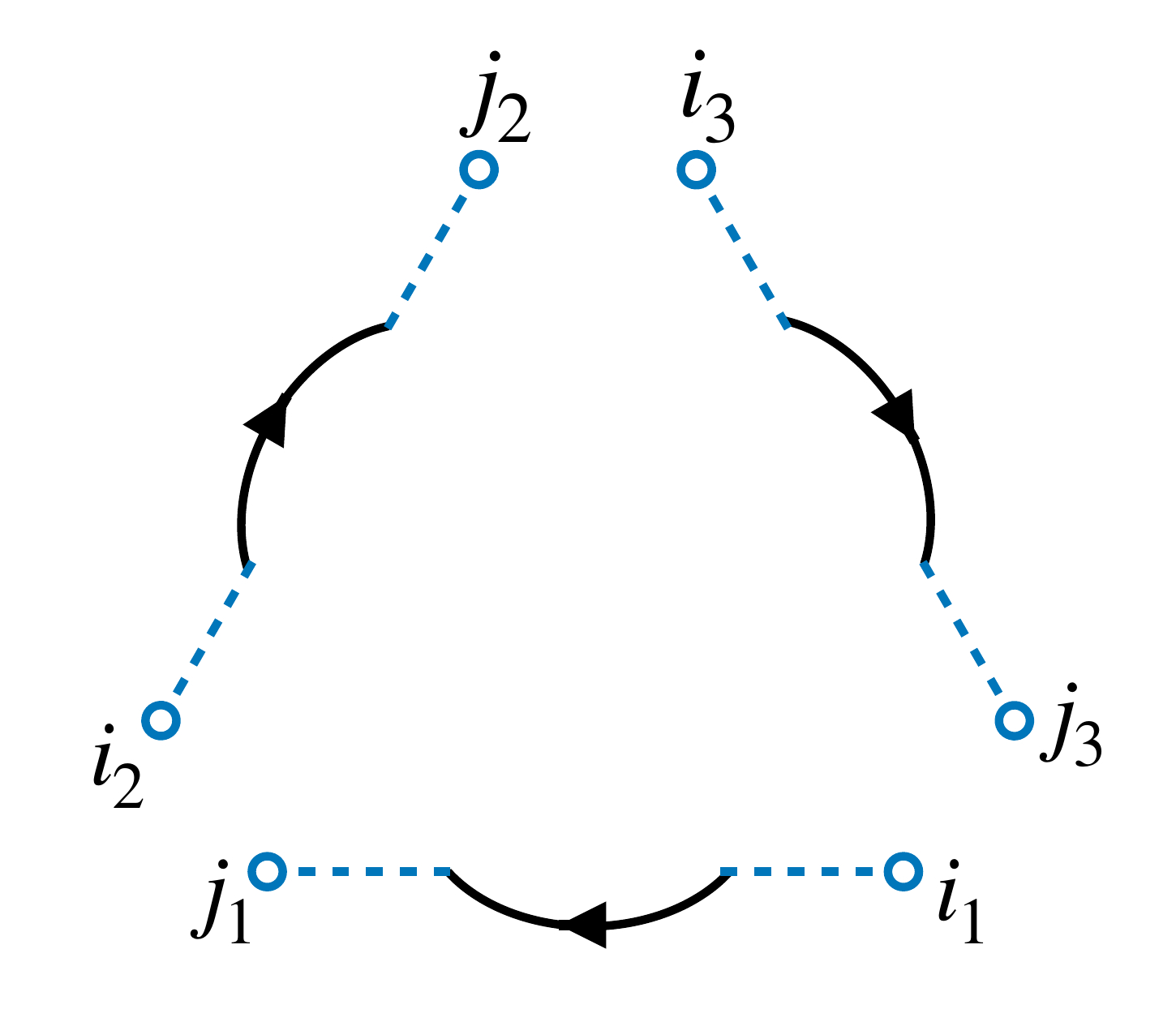}
  \caption{\centering}\label{fig:bc3}
\end{subfigure}%
\begin{subfigure}{.33\textwidth}
  \centering
  \includegraphics[width=.87\linewidth]{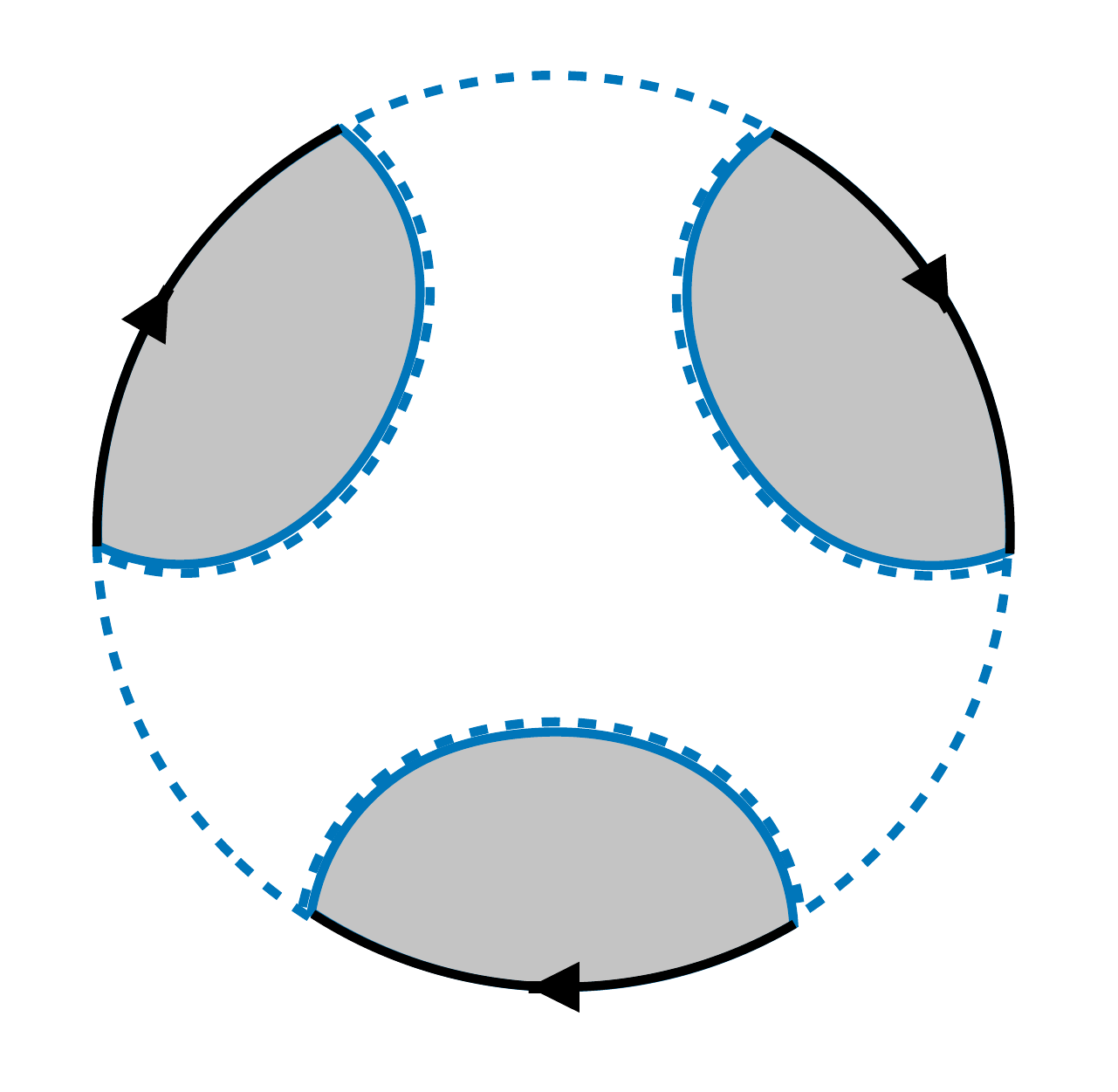}
   \caption{\centering}\label{fig:producttopology}
\end{subfigure}
\begin{subfigure}{.33\textwidth}
  \centering
  \includegraphics[width=.9\linewidth]{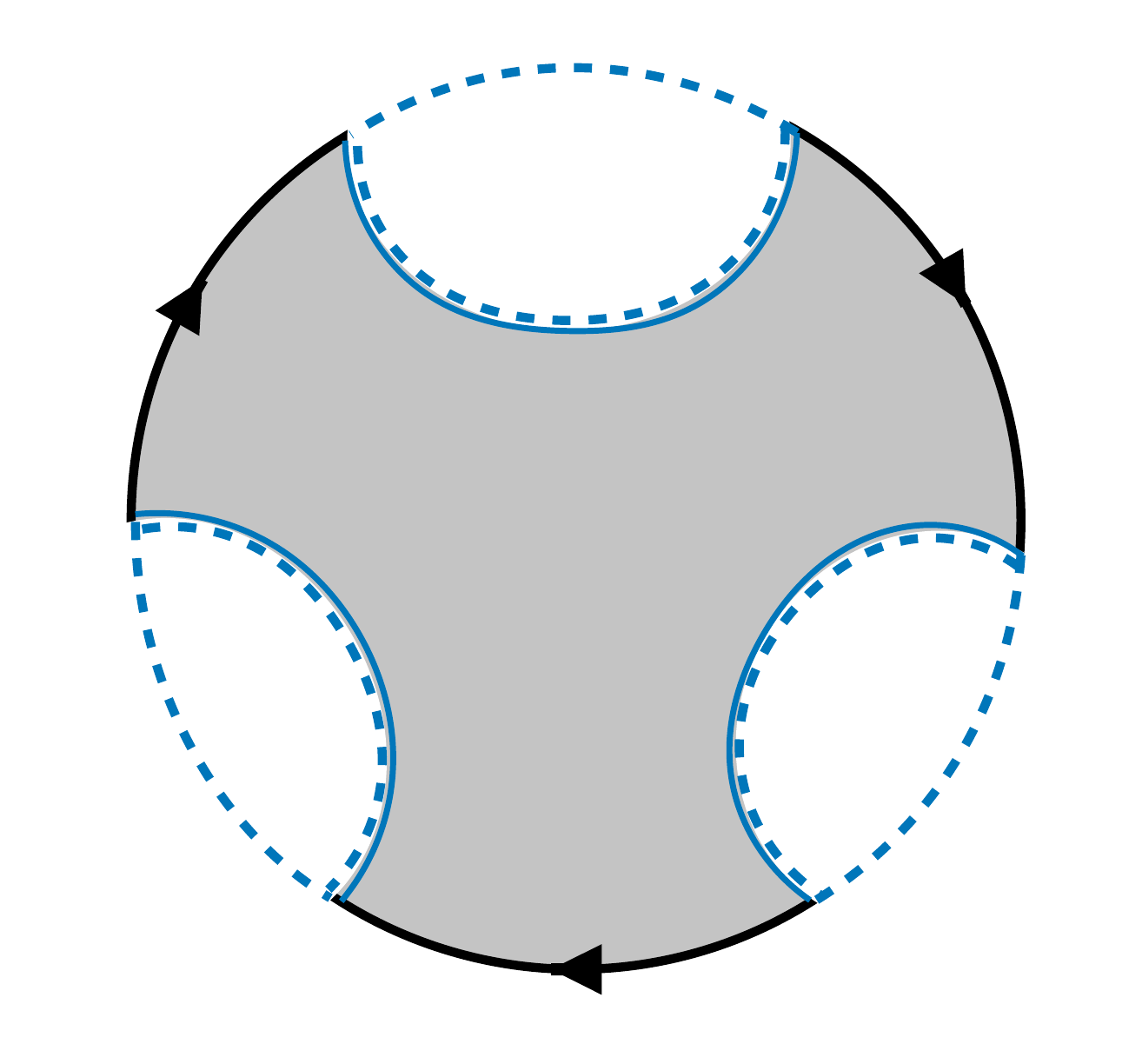}
   \caption{\centering}\label{fig:wormholetopology}
\end{subfigure}%
\caption{Gravitational path integral calculation for $S_n^{\mathrm{swap}}(R)_{\mathcal{B}}$. The boundary conditions are depicted by Fig.~\ref{fig:bc3}. Each solid black line corresponds to an asymptotic boundary with renormalised length $\beta$, and the dashed blue legs at their ends indicate the infinite past and infinite future of an EOW particle, with the index labelling the particle's internal state. The figure shows the case where $n=3$, which requires $3$ identical copies of these boundary conditions.  The contraction with the cyclic shift operator $\tau_{i_1 j_1 i_2 j_2 i_3 j_3} = \delta_{j_1,i_2} \delta_{j_2,i_3} \delta_{j_3, i_1}$ corresponds to a pairing of the asymptotic end points of the EOW worldlines that are drawn adjacent to each other. The path integral then has two dominant saddle geometries. The first, shown in Fig.~\ref{fig:producttopology}, has the topology of three disconnected discs. Since the internal state of each EOW particle is constant along its worldline (blue line), this leads to one single index loop (dashed line). The second saddle, shown in Fig.~\ref{fig:wormholetopology}, has a connected topology and three index loops.}
\label{fig:saddles2}
\end{figure}

To evaluate the entropy of the radiation~$R$, PSSY compute expression~\eqref{eq:replica_trick_2} for $S_n^{\mathrm{swap}}(R)_{\mathcal{B}}$ for any integer $n \geq 2$. To illustrate the corresponding path integral, which now involves $n$ copies of the spacetime boundary conditions described above, we consider the case where $n=3$ (see Fig.~\ref{fig:bc3}). The partition function $Z_n$ then admits the form
\begin{align} \label{eq:Z3partition}
    Z_3 = Z[\mathcal{B}^{\times 3}, \tau_3] 
    & = \sum_{i_1, j_1, i_2, j_2, i_3, j_3} \tau_{i_1 j_1 i_2 j_2 i_3 j_3}  \int_{\substack{i_1 \stackrel{\beta}{\longrightarrow} j_1 \\ i_2 \stackrel{\beta}{\longrightarrow} j_2 \\ i_3 \stackrel{\beta}{\longrightarrow} j_3}} \mathcal{D}g\mathcal{D}\psi e^{-I_{\mathrm{JT}}[g,\phi] - I_{\mathrm{particle}}[g]} \ .
\end{align}
The path integrals appearing in these expression can be evaluated using a saddle point approximation for the geometry~$g$. Remarkably, it has been found that there is more than one saddle. The topologies of the two dominant saddles are shown in Fig.~\ref{fig:producttopology} and~\ref{fig:wormholetopology}, respectively.\footnote{We omit sub-dominant saddles that have a topology which connects two out of the three boundaries.} The first is a factorised topology, in which each EOW brane intersects the same asymptotic boundary twice. The second topology, the \emph{wormhole solution}, has EOW branes that connect different boundaries. 

Since the JT action $I_{\mathrm{JT}}$ only depends on the topology, the integral on the right hand side of~\eqref{eq:Z3partition} can be expressed in terms of the partition function $\bar{Z}_n$ evaluated for a disc topology with $n$ asymptotic boundaries of length~$\beta$. In particular, for the topology shown in Fig.~\ref{fig:producttopology}, which consists of three disconnected discs, the integral evaluates to
\begin{align}
    \delta_{i_1,j_1} \delta_{i_2, j_2} \delta_{i_3,j_3} \bar{Z}_1^3 \ .
\end{align}
Similarly, the topology shown in Fig.~\ref{fig:wormholetopology} contributes a term of the form
\begin{align}
    \delta_{i_1,j_3} \delta_{i_2, j_1} \delta_{i_3,j_2} \bar{Z}_3 \ .
\end{align}
Inserting these in~\eqref{eq:Z3partition} and using that $\tau_{i_1 j_1 i_2 j_2 i_3 j_3} = \delta_{j_1,i_2} \delta_{j_2,i_3} \delta_{j_3, i_1}$, one finds
\begin{align}
    Z_3 =  \sum_{i_1, i_2, i_3} \delta_{i_1,i_2} \delta_{i_2, i_3} \delta_{i_3,i_1} \bar{Z}_1^3 + \delta_{i_1,i_1} \delta_{i_2, i_2} \delta_{i_3,i_3} \bar{Z}_3 
    = 2^k \bar{Z}_1^3 + (2^k)^3 \bar{Z}_3 \ .
\end{align}
Note also that
\begin{align}
    Z_1 = Z[\mathcal{B}, 1] = 2^k \bar{Z}_1 \ .
\end{align}
We now have all ingredients required for the evaluation of~\eqref{eq:replica_trick_2} for $n=3$. Returning to a general integer~$n$, the expression for the swap entropy of order~$n$ reads 
\begin{align}
    S_n^{\mathrm{swap}}(R)_{\mathcal{B}} 
  = \frac{1}{1-n}\log \frac{2^k \bar{Z}_1^n + 2^{nk} \bar{Z}_n}{2^{nk} \bar{Z}_1^n} 
  = \frac{1}{1-n}\log \Bigl(2^{(1-n) k} + \frac{\bar{Z}_n}{\bar{Z}_1^n} \Bigr) \ .
\end{align}
Hence, in a first approximation to JT gravity (see~\cite{penington2019replica} for more precise statements), where one only keeps the dependence on the $S_0$-term in the action~\eqref{eq:JTaction}, one finds that $\bar{Z}_n$ is proportional to $2^{\chi S_{\mathrm{BH}}}$, where $\chi$ is the Euler characteristic of the geometry, which in our case is $\chi=1$, and where $S_{\mathrm{BH}} \approx S_0$. We thus have
\begin{align}
    S_n^{\mathrm{swap}}(R)_{\mathcal{B}} 
    \approx \frac{1}{1-n}\log \Bigl(2^{(1-n) k} + \mathrm{const} \, 2^{(1-n) S_{\mathrm{BH}}} \Bigr) \ .
\end{align}
Consequently, for $S_{\mathrm{BH}} - k \gg 1$, corresponding to a young black hole, the first term in the bracket dominates, so that $S_n^{\mathrm{swap}}(R)_{\mathcal{B}} \approx k$. Conversely, for $k - S_{\mathrm{BH}} \gg 1$, corresponding to an old black hole, $S_n^{\mathrm{swap}}(R)_{\mathcal{B}} \approx S_{\mathrm{BH}}$. This conclusion may be summarised as
\begin{align} \label{eq:pathintegralresult}
    S_n^{\mathrm{swap}}(R)_{\mathcal{B}} \approx \min\{k, S_{\mathrm{BH}}\} \ .
\end{align}
By analytic continuation, the result can be extended to $n \to 1$, so that it also holds for $S^{\mathrm{swap}}(R)_{\mathcal{B}}$. This shows that the swap entropy follows the Page curve. 

A main feature of the path integral approach is that it does not require the definition of quantum states. In particular, \eqref{eq:pathintegralresult} was derived without any reference to a state space for the black hole spacetime. One may nonetheless try to understand the above conclusions in a picture where the black hole spacetime is regarded as a quantum system, which we denote by~$\tilde{B}$.\footnote{PSSY denote this system as~$B$~\cite{penington2019replica}. We use a tilde to distinguish it from the primary black hole interior described in Section~\ref{sec:toy}.} By construction, the index of the state $\ket{i}_R$ of the radiation system is encoded in the internal degree of freedom of the EOW particle. Assuming that the joint state of $R$ and $\tilde B$ is pure, it must have the form
\begin{equation}\label{eq:BW|R}
    \ket{\Psi}_{\tilde B R} = \frac{1}{\sqrt{2^k}}\sum_{i=1}^{2^k}\ket{\psi_i}_{\tilde B}\ket{i}_R\, ,
\end{equation}
where, $\ket{\psi_i}_{\tilde B}$ is a quantum state of $\tilde B$ in which the EOW particle admits state~$i$, for $i \in \{1, \ldots, 2^k\}$.

PSSY discuss this state in the light of the island formula~\eqref{eq:island}. Note that the formula features an entropy both on the left and on the right hand side. Depending on how one interprets the state~\eqref{eq:BW|R}, it can serve as a basis for the computation of either of these two entropies. We will use superscripts to make this distinction explicit. Specifically, the state $\ket{\Psi}^{\mathrm{grav}}_{\tilde B R}$ shall be such that the von Neumann entropy of $R$ matches the entropy on the left hand side of the island formula~\eqref{eq:island}. Conversely, $\ket{\Psi}^{\mathrm{bulk}}_{\tilde B R}$ is interpreted as the state that underlies the calculation of the von Neumann entropy on the right hand side of~\eqref{eq:island}.  

Let us start with the latter. The island $I$ may either be empty or contain the EOW particle, which lies in~$\tilde{B}$. In the second case, since the joint state of $\tilde{B}$ and $R$ is pure, we have 
\begin{align}
    S(R \cup I)_{\rho^{\mathrm{bulk}}}  = S(R \tilde{B})_{\rho^{\mathrm{bulk}}} = 0 \ ,
\end{align}
whereas the area term in~\eqref{eq:island} contributes $S_{\mathrm{BH}}$. Conversely, if the island is empty, we need to evaluate the von Neumann entropy of the radiation system $R$ alone. Since two states $\ket{\psi_i}^{\mathrm{bulk}}_{\tilde B}$, for different values of $i$, correspond to situations with distinguishable states of the EOW particle, one should treat them as mutually orthogonal, i.e., 
\begin{align} \label{eq:bulkoverlap}
    \braket{\psi_i}{\psi_j}^{\mathrm{bulk}}_{\tilde B} = \delta_{i j} \ .
\end{align}
Hence, when tracing over $\tilde B$, we get a maximally mixed state, so that
\begin{align} 
    S(R)_{\rho^{\mathrm{bulk}}} = k \ .
\end{align}
The island formula~\eqref{eq:island} thus yields
\begin{align} \label{eq:PSSYPage}
    S^{\mathrm{swap}}(R)_{\mathcal{B}} = \min\{S_{\mathrm{BH}},k\} \ ,
\end{align}
which is compatible with the result~\eqref{eq:pathintegralresult} of the path integral calculation. 

Let us now return to the state $\ket{\Psi}^{\mathrm{grav}}_{\tilde B R}$, which we defined as the state whose von Neumann entropy matches the entropy on the left hand side of~\eqref{eq:island}, i.e., the swap entropy $S^{\mathrm{swap}}(R)_{\mathcal{B}}$. For an old black hole, we have $S^{\mathrm{swap}}(R)_{\mathcal{B}} = S_{\mathrm{BH}} < k$, which immediately implies that the states $\ket{\psi_i^{\mathrm{grav}}}_{\tilde B}$ cannot be mutually orthogonal. Within the PSSY model, this is usually explained by the existence of ``microscopic quantities'', which give raise to random variables $R_{ij}$ with zero mean and unit variance, such that
\begin{align} \label{eq:gravoverlap}
    \braket{\psi_i}{\psi_j}^{\mathrm{grav}}_{\tilde B} = \delta_{i j} + 2^{-S_{BH}/2} R_{ij} \ .
\end{align}
Furthermore, to make this compatible with~\eqref{eq:bulkoverlap}, the latter is interpreted as an (implicit) ensemble average of~\eqref{eq:gravoverlap} over $R_{ij}$. 

Let us now try to relate the two claims about the overlap of the black hole states $\ket{\psi_i}_{\tilde{B}}$, \eqref{eq:bulkoverlap} and~\eqref{eq:gravoverlap}, to the ER model introduced in Section~\ref{sec:elusivereference}. Expression~\eqref{eq:bulkoverlap} may best be understood by identifying the system $\tilde B$ of the PSSY model with the joint system $W \bar{W} B$ of the ER model, which we called the \emph{effective interior} in Section~\ref{sec:toy}. More specifically, one may identify the states $\ket{\psi_i}^{\mathrm{bulk}}_{\tilde B}$ with the states $\ket{\psi_i}_{W\bar{W}B}$ defined in Proposition~\ref{pr:Schmidt}. The orthogonality statement~\eqref{eq:bulkoverlap} then coincides with the assertion of the proposition. 

Conversely, \eqref{eq:gravoverlap} may be understood within the ER model by interpreting $\tilde{B}$ as the \emph{primary interior} $B$ of the ER model. Specifically, the states $\ket{\psi_i}^{\mathrm{grav}}_{\tilde B}$ may be identified with the vectors $\ket{\psi_i^w}_B$ defined by Proposition~\ref{pr:Schmidt}. While these vectors are not necessarily normalised, they gain an operational meaning within the many-black-hole scenario, for their overlap~$\braket{\psi^w_i}{\psi^w_j}_B$ can be determined by an experiment, e.g., via quantum state tomography as described in Section~\ref{sec:toy}; see~\eqref{eq:internaloverlap}. In fact, one may as well  verify the modulus of the overlap more directly via a swap test. For this experiment, one would need to choose a number of pairs out of the collection of black holes with radiation states $\ket{i}_R$ and $\ket{j}_R$. This could be done by measuring the radiation $R$ in the basis $\{\ket{i}\}_{i=1, \ldots,  2^k}$ and post-selecting on the corresponding outputs. One could then apply a swap test to each pair, which yields an estimate for the squared overlap, 
\begin{align} \label{eq:swaptest}
  \frac{\bigl|\braket{\psi^w_i}{\psi^w_j}_B \bigr|^2}{\braket{\psi^w_i}{\psi^w_i}_B \braket{\psi^w_j}{\psi^w_j}_B} \ .
\end{align}
Note that this experiment, like the tomography experiment, can be carried out without access to the reference~$W$. 

To compare this to~\eqref{eq:gravoverlap}, we calculate the average  over~$W$ of the square of the scalar product occurring in~\eqref{eq:swaptest}. For this we use Lemma~3.5 of~\cite{dupuis2014one}, which yields
\begin{multline}
    \Bigl\langle\bigl| \braket{\psi^w_i}{\psi^w_j}_B \bigr|^2\Bigr\rangle_W = \frac{1}{|\mathcal{W}|} \sum_{w \in \mathcal{W}}  \bigl|\braket{\psi^w_i}{\psi^w_j}_B \bigr|^2
    = \frac{1}{|\mathcal{W}|} \sum_{w \in \mathcal{W}} \tr\bigl( \tau_{B B '} \ketbra{\psi^w_i}{\psi^w_i}_B \otimes \ketbra{\psi^w_j}{\psi^w_j}_{B'} \bigr) \\
    = \frac{|R|^2}{|\mathcal{W}| } \sum_{w \in \mathcal{W}} \tr \bigl( \tau_{B B'} \otimes \ketbra{i}{i}_R \otimes \ketbra{j}{j}_R U_w \ketbra{0}{0}_M U_w^{*} \otimes  U_w \ketbra{0}{0}_{M'} U_w^{*} \bigr) \\
    = |R|^2 \tr \bigl( \tau_{B B'} \otimes \bigl(\ketbra{i}{i}_R \otimes \ketbra{j}{j}_{R'} \, \alpha 1_{B  B'} \otimes 1_{R R'} + \beta \tau_{B B'} \otimes \tau_{R R'} \bigr) \bigr)
\end{multline}
with $\alpha, \beta \in \mathbb{R}$ such that $\alpha |B R|^2 + \beta |B R| = 1 $ and $\alpha |B R| + \beta |B R|^2 = 1$, and where $\tau_{B B'}$ denotes the operator that swaps the states of the two systems $B$ and $B'$. Hence, with $\alpha=\beta = \frac{1}{|B R|^2 + |B R|} = |BR|^{-2} + \O(\frac{1}{|BR|^3})$, we conclude that the average square overlap is
\begin{multline} \label{eq:woverlap}
    \Bigl\langle\bigl| \braket{\psi^w_i}{\psi^w_j}_B \bigr|^2\Bigr\rangle_W
    = \frac{1}{|B|^2} (|B| + |B|^2 \delta_{ij}) (1 + \O(| B R |^{-1})) \\
    = \delta_{ij} + 2^{-S_{\mathrm{BH}}} + \O(2^{-(S_{\mathrm{BH}}+k)})  \ ,
\end{multline}
where we used $|B| = 2^{S_{\mathrm{BH}}}$ and $|R| = 2^k$. This shows that the operational viewpoint within the ER model is indeed compatible with~\eqref{eq:gravoverlap}. 

\begin{table}
\begin{tabular}{p{3.35cm} p{1.7cm} p{3.5cm} p{4cm} p{0.8cm}}
    PSSY model & &  \multicolumn{3}{l}{ER model}  \\[-0.5ex]
    partner state to~$\ket{i}_R$ & & partner state to~$\ket{i}_R$ & relevant subsystem & size \\[0.1ex]\hline\\[-2.4ex]
    $\ket{\psi_i}^{\mathrm{bulk}}_{\tilde B}$ & $\longleftrightarrow$ & $\ket{\psi_i}_{W \bar{W} B}$ & effective interior $B W \bar{W}$ &  $\geq k$ \\[0.4ex]
    $\ket{\psi_i}^{\mathrm{grav}}_{\tilde B}$ & $\longleftrightarrow$ &
    $\ket{\psi^w_i}_B$ (for typical $w$) & primary interior $B$ & $S_{\mathrm{BH}}$
\end{tabular}
\caption{Dictionary between the PSSY model and the ER model. The states $\ket{\psi_i}_{\tilde{B}}$ defined within the PSSY model may be interpreted either as states of the bulk (whose entropy appears on the right hand side of the island formula~\eqref{eq:island}) or as states of the black hole itself (which relate to the entropy on the left hand side of the island formula). Depending on this interpretation, they correspond in the ER model to states of either the effective or the primary black hole interior.  \label{tb:dictionary}}
\end{table}

In summary, we have thus arrived at a dictionary that allows us to identify the black hole states in the PSSY model with corresponding states in the ER model as shown in Table~\ref{tb:dictionary}.  If one interprets the partner states to the radiation states $\ket{i}_R$ in the PSSY model as gravitational states then they have a non-trivial overlap. In the ER model, they span the Hilbert space $B$ of the primary interior. For an old black hole, the dimension $|B| = 2^{S_{\mathrm{BH}}}$ of this space may be arbitrarily smaller than the dimension $2^k$ of the radiation field.  Conversely, if one interprets the partner states to the radiation states $\ket{i}_R$ in the PSSY model as bulk states then they are mutually orthogonal and thus span a space of dimension $2^k$. In the ER model, this space corresponds to the effective interior $B W \bar{W}$, which even for an old black hole is large enough to carry them. 

To conclude this section, we note that the above dictionary provides some intuition on the ``location'' of the EOW particle and thus of the island~$I$ in the island formula~\eqref{eq:island}. For an old black hole, $I$ includes the bulk field beyond the event horizon and hence, in the PSSY model, the  EOW particle. But does this mean that the EOW particle is located at the outside of the black hole? The ER model provides a clear answer to this question. Since $S_{\mathrm{BH}} < k$, the EOW particle internal state can impossibly be encoded in the primary interior~$B$. Crucially, however, it is also not contained in the radiation field~$R$ of a single black hole. Rather, it is encoded in the effective interior $B W \bar{W}$, which includes the reference. And since this reference can be retrieved by tomography applied to many black holes, one may say that the EOW particle is (partially) contained in their joint radiation field.

\subsection{The Page curve for a black hole in superposition}\label{sec:example}

As a consistency check, we consider the concrete example of a black hole that is in superposition of two stages of the radiation process. More precisely, we imagine a situation that arises when starting from a superposition of black holes of different initial sizes. In order to maintain the controlled semiclassical gravitational path integral calculation and avoid handling a  superposition of macroscopically distinct metrics, we superpose them at different stages of the evaporation such that their Bekenstein-Hawking entropy~$S_{\mathrm{BH}}$ is equal whereas the radiation in one branch contains more quanta than the other. Therefore, both branches share the same spacetime geometry and only the radiation system is in a superposition of two different sizes.

The ER model provides a rule for defining the joint state of the black hole and the radiation system. We define this state in terms of two branches, each of which is given by an expression of the form~\eqref{eq:BRWpurification}. The first shall have $k$ radiation qubits. Using Proposition~\ref{pr:Schmidt}, we may write this branch as\footnote{To simplify the notation, we take $W$ from now on to denote the joint system that includes $\bar{W}$. This can be done without loss of generality because, according to~\eqref{eq:BWstates}, the state on $W \bar{W}$ is supported on the subspace spanned by $\ket{w}_W \ket{w}_{\bar{W}}$.}
\begin{equation}
    \ket{\Psi^0}_{WBR} := \frac{1}{\sqrt{|\mathcal{W}|}}\sum_{w \in \mathcal{W}} \ket{w}_WU_w \ket{0}_M=\frac{1}{\sqrt{2^k}}\sum^{2^k}_{i=1}\ket{\psi_i}_{WB}\ket{i}_R \ ,
\end{equation}
where $B$ is a system whose size corresponds to the Bekenstein-Hawking entropy, i.e.,  $S_{BH}=\log|B|$. The second branch involves the same system $B$ but shall have $m$ radiation qubits,
\begin{equation}
    \ket{\Psi^1}_{WBR} := \frac{1}{\sqrt{|\mathcal{W}|}}\sum_{w \in \mathcal{W}} \ket{w}_WU_w \ket{1}_M=\frac{1}{\sqrt{2^m}}\sum^{2^m}_{i=1}\ket{\varphi_i}_{WB}\ket{i}_R \ .
\end{equation}
One should think of $\ket{0}_M$ and $\ket{1}_M$  as the initial states of collapsing matter shells with distinct masses, which means that they are mutually orthogonal. We consider two possible superposed initial states, $\ket{\pm}_M:=\frac{1}{\sqrt{2}}(\ket{0}\pm\ket{1})_M$. By linearity, the resulting black hole states are  
\begin{equation}
    \ket{\Psi^\pm}_{WBR}:=\frac{1}{\sqrt{2}}\ket{\Psi^0}\pm\frac{1}{\sqrt{2}}\ket{\Psi^1}=\frac{1}{\sqrt{|\mathcal{W}|}}\sum_{w \in \mathcal{W}} \ket{w}_WU_w \ket{\pm}_M \ .
\end{equation}
Our goal is to evaluate the conditional entropy $S(R|W)_{\Psi^+}$ for one of the states, $\ket{\Psi^+}_{WBR}$.

\begin{figure}
    \centering
    \includegraphics[width=0.6\linewidth]{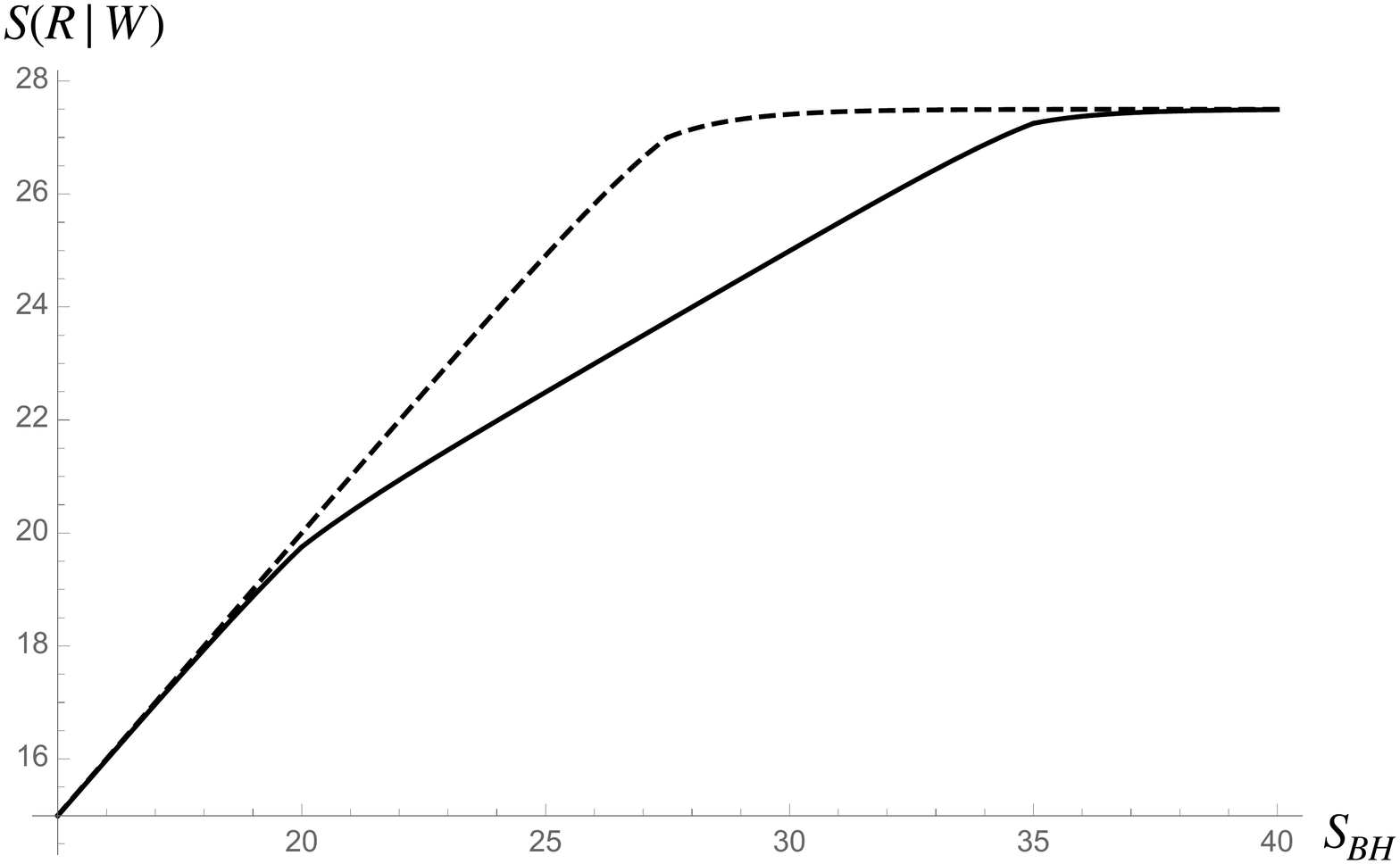}
    \caption{Page curve for a black hole in a superposition state. The solid curve shows the radiation entropy $S(R|W)$, calculated within the ER model, for a situation where the number of emitted radiation qubits is a fixed superposition between $k=35$ and $m=20$, depending on the size~$S_{\mathrm{BH}}$ of the black hole. It matches well with the curve derived using gravitational path integrals in~\cite{wang2021refined}. Note that this curve has no transition at the Page time, where $S_{\mathrm{BH}}$ equals $S(R) \approx 28$. For comparison, the standard Page curve with a transition at the Page time is indicated by the dashed line. }
    \label{fig:page2}
\end{figure}

We use the fact that the above states may be obtained from an initial state where the matter shell is entangled with an auxiliary qubit $A$,
\begin{equation}
    \frac{1}{\sqrt{2}}(\ket{0}_M\ket{0}_A+\ket{1}_M\ket{1}_A) \ ,
\end{equation}
by measuring the system $A$ with respect to the basis $\{\ket{\pm}_A = \frac{1}{\sqrt{2}}(\ket{0} \pm \ket{1})_A\}$. Suppose that the outcome of this measurement, which we label by $+$ or $-$, is stored  in a classical register~$C$. After the time evolution prescribed by the unitary~$U_w$, the total system without~$A$ thus admits a state of the form (now written as a density operator)
\begin{equation}
    \Psi_{WBRC}:=\frac12\Psi^+_{WBR}\otimes\ketbra{+}{+}_C + \frac12\Psi^-_{WBR}\otimes\ketbra{-}{-}_C
\end{equation}
Note that the marginal state obtained by tracing over~$C$ has the form
\begin{equation}
    \Psi_{WBR}:=\frac12\Psi^+_{WBR} + \frac12\Psi^-_{WBR} = \frac12 \Psi^0_{WBR} + \frac12 \Psi^1_{WBR} \ .
\end{equation}
Using this decomposition, it is easy to see that the entropies of the components of these states are related by
\begin{equation}\label{eq:infotheoretic}
  \frac12 S(R|W)_{\Psi^+} + \frac12 S(R|W)_{\Psi^-}
 = S(R|WC)_\Psi
 \approx
   S(R|W)_\Psi
\approx 
  \frac12 S(R|W)_{\Psi^0} + \frac12 S(R|W)_{\Psi^1} \ , 
\end{equation}
where the term after the first approximation sign can at most be $1$~bit larger than the left hand side. Similarly, the term after the second approximation sign can at most be $1$~bit smaller than the term preceding it. The overall approximation thus holds up to $1$~bit. 

Since $\{U_w\}_{w\in\mathcal{W}}$ is a unitary $N$-design, we may approximate the sum over $w$ in~$S(R|W)_{\Psi^\pm}$ by an integral over all isometries $U$ from $M$ to $B R$ with respect to the Haar measure~$\dd U$. Hence, up to an error that is arbitrarily small for large enough~$N$, we have
\begin{align} \label{eq:approxHaar}
    S(R|W)_{\Psi^\pm}\approx \int \dd U \, S(R)_{U\ketbra{\pm}{\pm} U^{*} } \ .
\end{align}
Let now $Z$ be the unitary on~$M$ that flips $\ket{\pm}_M$ to $\ket{\mp}_M$. Since the Haar measure $\dd U$ is invariant under the right action of $Z$, we find that the right hand side of the approximation~\eqref{eq:approxHaar} is independent of whether we consider $\Psi^+$ or $\Psi^-$.  Hence, the two terms on the left hand side of~\eqref{eq:infotheoretic} are approximately equal, and we find
\begin{equation}
    S(R|W)_{\Psi^+} \approx \frac12 S(R|W)_{\Psi^0} + \frac12 S(R|W)_{\Psi^1}\ .
\end{equation}

We have thus shown that, within the ER model, the entropy $S(R|W)_{\Psi^+}$ of the radiation of a black hole in the superposed state $\ket{\Psi^+} = \frac{1}{\sqrt{2}} ( \ket{\Psi^0} + \ket{\Psi^1})$ is well approximated by the average of the entropy of each branch. The latter is given by the standard Page curve as discussed in Section~\ref{sec:toy}.  Fig.~\ref{fig:page2} shows a plot of~$S(R|W)_{\Psi^+}$ for the case where $k=35$ and $m=20$.\footnote{For an arbitrary entanglement spectrum, the transition point of the Page curve (as shown in Fig.~\ref{fig:page2}) is no longer characterized by the von Neumann entropy $S(R)$ being approximately equal to~$S_{\mathrm{BH}}$, but rather is refined to two transition points characterized by the min/max-entropy $S_\mathrm{min/max}(R)$. This has been shown in~\cite{akers2021leading,wang2021refined}.} It exactly reproduces the curve that was obtained via gravitational path integral calculations within the PSSY model~\cite{wang2021refined} (see Fig.~7 of that reference).

\section{Discussion} \label{sec:discussion}

The quantum de Finetti theorem, combined with the recent calculations based on gravitational path integrals and applied to a setup consisting of many identically prepared black holes as described in Section~\ref{sec:results}, suggests the following explanation for the black hole information puzzle. The radiation $R_1$ of a single black hole, observed independently of the other black holes, is thermal as predicted by Hawking and thus has a large von Neumann entropy~$S(R_1)$. In particular, $S(R_1)$ will after the Page time exceed the Bekenstein-Hawking entropy~$S_{\mathrm{BH}}$ of the black hole. However, after the Page time, the  radiation fields $R_1 \cdots R_N$ of a collection of black holes get correlated. Consequently, the regularised entropy $\lim_{N \to \infty} \frac{S(R_1 \cdots R_N)}{N}$ will become strictly smaller than~$S(R_1)$. It is this regularised entropy that follows the Page curve. Hence, Hawking's ever growing radiation entropy and the Page curve  are compatible with each other.

According to the \emph{Elusive Reference (ER)} model introduced in Section~\ref{sec:model}, the correlation between the radiation fields $R_1, \ldots, R_N$ can be understood analogously to a situation where we would have $N$ identically oriented spins whose direction is defined relative to a reference frame~$W$, which may however not be directly observable. If we see only one single spin, its orientation looks random and thus has maximum entropy. Conversely, the collection of all $N$~spins is highly ordered, as they are all pointing in the same direction. The regularised entropy of the collection thus approaches zero. In the ER model, the spins correspond to the quantum fields that carry the Hawking radiation, for instance. Accordingly, the reference~$W$, rather than defining directions in real space, defines directions in the Hilbert space that underlies the description of these quantum fields. Hence, if we see radiation~$R$ from a single black hole but are lacking the reference, $R$~looks random to us, i.e., its von Neumann entropy~$S(R)$ is large.

The ER model also reconciles apparently contradicting conclusions that one could have drawn about the interior of a black hole.  In fact, within the ER model there are two different types of black hole interiors. One of them, the \emph{primary interior}, only has $S_{\mathrm{BH}}$ degrees of freedom and thus shrinks as the black hole evaporates. It corresponds to the black hole interior that is commonly considered for the derivation of the Page curve. The other, the~\emph{effective interior}, includes the reference~$W$ and thus remains large enough to contain all partner modes of the radiation even for an old black hole. It thus corresponds to the black hole interior as considered by Hawking.  

We hope that the light that the quantum de Finetti theorem sheds on black holes helps resolving the disagreement between proponents of the recent replica wormhole calculations~\cite{penington2019replica,almheiri2020replica,almheiri2020entropy} and those who uphold a remnant picture, such as the one implied by loop quantum gravity~\cite{perez2015no,bianchi2018white,martin2019evaporating,rovelli2019subtle,rovelli2020end}. The tension seems to be focused on the \emph{central dogma}, which asserts that the number of degrees of freedom of a black hole as observed from the outside is  $S_{\mathrm{BH}}=\mathrm{Area}/4G_N$~\cite{almheiri2020entropy}. Advocates of the remnant picture make the objection that the central dogma contradicts with the vast volume behind the horizon throat~\cite{christodoulou2015big}, also known as the \emph{bag of gold} geometry, that is well capable of accommodating the degrees of freedom that purify the Hawking radiation~\cite{perez2015no,rovelli2019subtle}. In the case of loop quantum gravity, these may be quantised geometric degrees of freedom at the Planck scale supported on the big volume behind the horizon. Alternatively, it has been proposed that the hidden degrees of freedom may be identified with baby universes that emerge and interact behind the horizon~\cite{coleman1988black,giddings1988loss,polchinski1994possible,marolf2020transcending,marolf2020observations}.

The disagreement about the central dogma may be resolved simply by associating these hidden degrees of freedom to the reference~$W$. The bag of gold would then correspond to the effective black hole interior, which we defined as the primary interior together with the reference~$W$. This large interior, which thus ``hides'' the reference, is compatible with Hawking's ever growing entropy~$S(R)$. Conversely, the primary interior is relevant --- and the central dogma thus valid --- in all considerations that are to be understood ``relative to'' the reference. In the case of statements about entropies, this means that one conditions on~$W$. As we have argued here, the replica wormhole calculations of the radiation entropy do exactly this, i.e., they compute~$S(R|W)$. Hence, it may well be that the disagreement around the replica wormhole calculations is simply due to a mismatch of quantities, $S(R)$ and $S(R|W)$, that one is talking about, and which refer to different notions of the black hole interior. 

The results presented here are also compatible with MM's arguments for the existence of superselection sectors induced by the baby universe Hilbert space, which is not accessible to asymptotic observers~\cite{marolf2020transcending,marolf2020observations,marolf2021page}. We think of the existence of superselection sectors more generally as a consequence of the missing access to a reference $W$ for the black hole dynamics, but one may as well regard $W$ as a carrier of the $\alpha$-states of the baby universe Hilbert space.  

In AdS/CFT, wormholes provide the essential motivation for the idea of emergent spacetime. They are the geometric avatars of correlations in gravity. Motivated by the close connection between (bipartite) entanglement and wormholes~\cite{maldacena2003eternal,van2010building}, it is even conjectured that wormholes like Einstein-Rosen bridges are universally equivalent to  Einstein-Podolsky-Rosen pairs (ER=EPR)~\cite{maldacena2013cool}. However, it is important to emphasize that wormholes that are supposed to represent entanglement are spatial ones that are not localized in time, whereas replica wormholes are spacetime wormholes which are localized both in space and in time. While the difference between them is explored in~\cite{marolf2020transcending,marolf2020observations}, the de Finetti theorem confirms that the replica wormholes that appear in the recent calculations of the radiation entropy, as opposed to the spatial wormholes commonly studied, only represent \emph{classical correlations}, i.e., there is no entanglement being mediated by them. We anticipate a rich correspondence between different types of wormholes in gravity and different types of corresponding correlations.

There are important questions that we did not address, such as what the reference $W$ represents exactly and how we should describe it from the viewpoint of an observer falling into the black hole~\cite{almheiri2013black,marolf2013gauge}. It could well be that semi-classical notions of spacetime, such as a smooth horizon, depend on whether or not one takes into account~$W$. In candidate theories of quantum gravity, one could make more detailed speculations about~$W$. For example, in holographic theories with a gravity/ensemble duality~\cite{saad2019jt,saad2018semiclassical,stanford2019jt}, $W$~could represent the distinct boundary Hamiltonians for an ensemble of quantum field theories; in the fuzzball paradigm, $W$~could represent the stringy and braney excitations that live on the horizon~\cite{mathur2005fuzzball,mathur2019resolving}; and loop quantum gravity would assign $W$ to the quantised geometry at the Planck scale~\cite{ashtekar2018quantum,bianchi2018white,martin2019evaporating,rovelli2020end}. Black holes still have lots of insights to offer for years to come. 

\section*{Acknowledgments}

We thank Laura Burri, Esteban Castroz-Ruiz, Ladina Hausmann, Giulia Mazzola, Joe Renes and Henrik Wilming for inspiring discussions. This work has been supported by the Swiss National Science Foundation via the Centers for Excellence in Research QSIT and SwissMAP, as well as project No.~200021\_188541. RR also acknowledges the hospitality of the Kavli Institute for Theoretical Physics at the UC Santa Barbara during the Quantum Physics of Information program.

\appendix

\section{Uniform convergence of R\'enyi entropies} \label{app:Renyiconvergence}

Any quantum-classical state of the form $\rho_{R W} = \sum_{w \in \mathcal{W}} p_w \rho_{R|w} \otimes \ketbra{w}{w}_W$ may be interpreted as follows. Given the classical information $W=w$, which occurs with probability $p_w$, the state on $R$ is $\rho_{R|w}$. This defines in a unique way a probability measure $\dd\rho$ on the space $\mathcal{S}(\h_R)$ of states on~$R$, 
\begin{align} \label{eq:discretemeasure}
    \dd\rho(s) = \sum_{\substack{w \in \mathcal{W} : \\ \rho_{R|w} \in s }}  p_w \qquad \forall \, s \subset \mathcal{S}(\h_R) \ ,
\end{align}
such that
\begin{align}
  \rho_{R W} = \int \dd\rho\, \rho_{R|w} \otimes \ketbra{w}{w}_W \ .
\end{align}
For the following, we consider a sequence $\{\rho^{(N)}\}_{N \in \mathbb{N}}$ of quantum-classical states $\rho^{(N)} = \rho_{R W}^{(N)}$ of the form above, and denote by $\{\dd\rho^{(N)}\}_{N \in \mathbb{N}}$ the corresponding sequence of probability distributions on $\mathcal{S}(\h_R)$. We say that this sequence \emph{converges in distribution} if the expectation value with respect to $\dd\rho^{(N)}$ of any real continuous and bounded function~$f$ on $\mathcal{S}(\h_R)$, i.e., $\int \dd\rho^{(N)} f(\rho)$, converges in the limit $N \to \infty$. 

\begin{lem} \label{lem:distributionconvergence}
  If the sequence $\{\dd\rho^{(N)}\}_{N \in \mathbb{N}}$ is convergent in distribution then the sequence $\{g_N\}_{N \in \mathbb{N}}$ of functions $g_N: n \mapsto 2^{(1-n) S_n(R|W)_{\rho^{(N)}}}$, for $N \in \mathbb{N}$, converges uniformly in any compact region of the complex half plane $\Re(n) \geq 1$.
\end{lem}

\begin{proof}

Since $\dd\rho^{(N)}$ is related to $\rho^{(N)}_{R W}$ via~\eqref{eq:discretemeasure}, we have 
\begin{multline} \label{eq:Snaverage}
      \qquad 2^{(1-n) S_n(R|W)_{\rho^{(N)}}} 
      =\sum_{w \in \mathcal{W}} p_w^{(N)} 2^{(1-n) S_n(R)_{\rho_{w}^{(N)}}} \\
      = \int \dd\rho^{(N)} 2^{(1-n) S_n(R)_\rho} 
      = \int \dd\rho^{(N)} \tr(\rho^n)  \ ,
\end{multline}
where we used~\eqref{eq:qcconditionalRenyi} for the first equality. 

Let $\eps > 0$. Since the space of states $\mathcal{S}(\h_R)$ (equipped with the trace norm $\| \cdot \|_1$) is compact, there must exist a continuous function $c^\eps$ from $\mathcal{S}(\h_R)$ to the space of probability distributions on a finite subset $\mathcal{V}$ of $\mathcal{S}(\h_R)$ such that any $\nu \in \mathcal{V}$ and any $\rho \in \mathcal{S}(\h_R)$ with $c^\eps(\rho)_\nu > 0$ are $\eps$-close. In the following, we will call a function $c^\eps$ with this property an $\eps$-coarse-graining of $\mathcal{S}(\h_R)$. 

Using the bound\footnote{The bound can be obtained by extending Lemma~2 of~\cite{raggio1995properties} to complex $n$ with $\Re(n)\ge 1$. }
 \begin{align*}
        | \tr(\rho^n) - \tr(\tilde{\rho}^n) | 
        \leq |n| \| \rho - \tilde{\rho} \|_1 \ ,
 \end{align*}
 which holds for any density operators $\rho$ and $\tilde{\rho}$ and $n$ with $\Re(n) \geq 1$, we have 
   \begin{align} \label{eq:approximationone}
       \int \dd\rho^{(N)} \tr(\rho^n)
       \stackrel{n\eps}{\approx} 
       \int \dd\rho^{(N)} \sum_{\nu \in \mathcal{V}} c^\eps(\rho)_{\nu} \tr(\nu^n) 
       = \sum_{\nu \in \mathcal{V}} p^{(N)}_\nu \tr(\nu^n) \ ,
   \end{align}
   where we defined the probability distribution $p^{(N)}$ on $\mathcal{V}$ by  $p^{(N)}_\nu = \int \dd\rho^{(N)} c^{\eps}(\rho)_\nu$. 
   
   By assumption, the sequence $\{\dd\rho^{(N)}\}_{N \in \mathbb{N}}$  of probability measures on  $\mathcal{S}(\h_R)$ converges in distribution. Hence, for any coarse-graining $c^\eps$, the sequence $\{p^{(N)}_\nu\}_{N \in \mathbb{N}}$  converges for any $\nu \in \mathcal{V}$. Since the set $\mathcal{V}$ is finite, there must exist $N_0$ and a probability distribution $p$ on $\mathcal{V}$ such that for all $N \geq N_0$ we have $\sum_{\nu \in \mathcal{V}} |p^{(N)}_{\nu} - p_\nu| \leq \eps$. This implies 
  \begin{align} \label{eq:approximationtwo}
     \sum_{\nu \in \mathcal{V}} p^{(N)}_\nu \tr(\nu^n)
       \stackrel{\eps}{\approx} 
      \sum_{\nu \in \mathcal{V}} p_\nu \tr(\nu^n) \ . 
  \end{align}
  Choose now any compact region of the complex plane $\Re(n) \geq 1$ and let $n_{\max}$ be the maximum modulus of~$n$ within this region. Then, for any $\tilde{\eps} > 0$, we may choose $\eps = \tilde{\eps}/(n_{\max}+1)$ and concatenate~\eqref{eq:Snaverage}, \eqref{eq:approximationone}, and~\eqref{eq:approximationtwo} to conclude that, for $N$ sufficiently large,
  \begin{align}
     2^{(1-n) S_n(R|W)_{\rho^{(N)}}} 
     \stackrel{\tilde{\eps}}{\approx}  
    \sum_{\nu \in \mathcal{V}} p_\nu \tr(\nu^n)       
  \end{align}
  holds for any $n$ within the region. Since the approximation is independent of $n$, we have thus shown that the convergence is uniform. 
\end{proof}

\section{Uniqueness of the conditional entropy} \label{app:conditionalentropyuniqueness}

For the construction of the reference $W$ in the $N$-black-hole scenario described in Section~\ref{sec:reference}, we applied the quantum de Finetti theorem, Theorem~\ref{thm:definetti_state}, to the total state of all radiation systems $R_1, \ldots, R_N$. However, we may as well apply Theorem~\ref{thm:definetti_state} to a subset consisting of only $M \leq N$ radiation systems. This yields an approximation for $n < M$ radiations systems in terms of a state of the form 
\begin{equation}\label{eq:alternativedefinettistate}
    \rho^{(N,M)}_{R_1\cdots R_n W}
    = \sum_w p_w^{(N,M)}\rho_{R|w}^{\otimes n}\otimes\ketbra{w}{w}_W \ , 
\end{equation}
which equals~\eqref{eq:qcstate} for $M=N$, but which may be different for different values of $M$. The following lemma asserts, however, that this difference is irrelevant when considering the conditional entropy of any single radiation system $R$. More precisely, in the limit of large $N$ and $M$, the conditional entropy $S(R|W)_{\rho^{(N,M)}}$ evaluated for these states is identical to the entropy $S(R|W)$ defined by~\eqref{eq:condW}, which corresponds to the case where $M=N$. 

\begin{lem}\label{lem:uniqueness}
For fixed boundary conditions $\mathcal{B}$ and for any $M,N \in \mathbb{N}$ with $M \leq N$, let $\rho^{(N,M)}$ be the state of the form~\eqref{eq:alternativedefinettistate} obtained by applying the quantum de Finetti theorem to $M$ radiation fields within the $N$-black hole system defined by $\mathcal{B}^{\times N}$. Then
   \begin{align}
   \lim_{M \to \infty} \lim_{N \to \infty} S(R|W)_{\rho^{(N,M)}} = S(R|W) \ .
\end{align}
\end{lem}

\begin{proof}
  We will use the fact that the von Neumann entropy $S(R)_{\sigma}$ of a state $\sigma_R$ can be estimated by applying the Empirical Young Diagram algorithm~\cite{keyl2001estimating,acharya2019measuring}, which requires as input $m$ systems prepared in the product state $\sigma_R^{\otimes m}$. The estimate converges towards the value $S(R)_{\sigma}$ in the limit $m \to \infty$. Since the output of the algorithm is a real number, we may write the action of the algorithm as an observable $\Theta_{R_1 \cdots R_m}$.\footnote{Specifically, $\Theta_{R_1 \cdots R_m}$ is a linear combination of Young Diagram projectors with the measurement outcomes corresponding to the empirical estimate of the entropy.} A quantitative analysis of the algorithm then shows that 
  \begin{align} \label{eq:entropyestimate}
    S(R)_{\sigma} \stackrel{\eps'}{\approx} \tr(\Theta_{R_1 \cdots R_m} \sigma_R^{\otimes m}) 
  \end{align}
  where $\eps' = \O(m^{-1/3})$. 

  Let now $M \leq N$ and $m \leq \sqrt{M}$. Since both $\rho^{(N)}_{R_1 \cdots R_n}$ and $\rho^{(N,M)}_{R_1 \cdots R_n}$ are  defined as de Finetti approximations of the same marginal state of a system consisting of $N$ black holes, we have
  \begin{equation}\label{eq:lemstep1}
    \sum_w p_w^{(N,M)} \rho_{R|w}^{\otimes m} 
      \stackrel{\eps''}{\approx} \sum_w p_w^{(N)} \rho_{R|w}^{\otimes m} 
  \end{equation}
  where $\eps'' = \O(M^{-1/4})$. Applying the entropy measurement $\Theta_{R_1 \cdots R_m}$ to both sides of~\eqref{eq:lemstep1}, noting that $\|\Theta_{R_1 \cdots R_m}\|_\infty \leq \log |R|$, and using the approximation~\eqref{eq:entropyestimate} on both sides, we find by linearity
  \begin{multline}
      S(R|W)_{\rho^{(N,M)}} = \sum_{w} p_w^{(N,M)} S(R)_{\rho_{w}} \stackrel{\eps'}{\approx}
      \sum_{w} p_w^{(N,M)} \tr(\Theta_{R_1 \cdots R_m} \rho_{R|w}^{\otimes m}) \\
      = \tr\Bigl(\Theta_{R_1 \cdots R_m} \sum_{w} p_w^{(N,M)} \rho_{R|w} ^{\otimes m}\Bigr)
      \stackrel{\eps'' \log |R|}{\approx} \tr\Bigl(\Theta_{R_1 \cdots R_m} \sum_{w} p_w^{(N)} \rho_{R|w}^{\otimes m}\Bigr)
      \stackrel{\eps'}{\approx} S(R|W)_{\rho^{(N)}} \ .
  \end{multline}
  We may choose $m = \lfloor \sqrt{M} \rfloor$, in which case the overall approximation is $2 \eps' + \eps'' \log |R| = \O(M^{-1/6})$. By assumption, the $N \to \infty$ limit of the right hand side equals $S(R|W)$. We can thus conclude that
  \begin{align}
  \limsup_{N \to \infty} \bigl| S(R|W)_{\rho^{(N,M)}} - S(R|W) \bigr| = \O(M^{-1/6}) \ .
  \end{align}
  This immediately yields the desired claim.
\end{proof}

\section{Proof of Theorem~\ref{thm:definetti_state}} \label{app:deFinetti}

The variant of the quantum de Finetti theorem, Theorem~\ref{thm:definetti_state}, which we employ to derive our main claims, differs in various aspects from the ones proved in the literature. For this reason we provide a separate proof here. 

A main ingredient to the proof is a POVM on the symmetric subspace $\mathrm{Sym}^k(\h_S)$ of a product space $\h_S^{\otimes k}$. For this we use a quantum $k$-design, which is defined as a finite family of states $\{\psi_{S|w}\}_{w \in \mathcal{W}}$ on $\h_S$ such that
\begin{equation}
   \frac{1}{|\mathcal{W}|} \sum_{w\in\mathcal{W}} \psi_{S|w}^{\otimes k} = \int \dd\psi\,\psi_S^{\otimes k} \ ,
\end{equation}
where the integration measure $\dd\psi$ on the right hand side is the Haar measure on states on~$\h_S$. Note that a quantum $k$-design is also a quantum $t$-design for any $t\le k$. Then the family 
\begin{align} \label{eq:povm}
    \{C_k\psi_{S|w}^{\otimes k}/|\mathcal{W}|\}_{w \in \mathcal{W}} \ ,
\end{align}
with $C_{k}:=\binom{k+\dim(\h_S)-1}{k}$  the dimension of the symmetric subspace $\mathrm{Sym}^k(\h_S)$, forms a POVM on that subspace. 

Another key ingredient to the proof is the following well-known statement, which asserts that permutation-invariant states have a symmetric purification.

\begin{lem}[Lemma~II.5 and Lemma~II.5' in \cite{christandl2007one}]\label{lem:purification}
  Let $\rho$ be a state on $\h_R^{\otimes N}\otimes\h_E$ that is per\-mu\-ta\-tion-invariant relative to $\h_E$. Then there exists a purification of $\rho$ in $\mathrm{Sym}^N({\h_R\otimes\h'_R})\otimes\h_E\otimes\h'_E$ with $\h_R\cong\h'_R$ and $\h_E\cong\h'_E$.
\end{lem}

We now state a more explicit version of Theorem~\ref{thm:definetti_state}, from which the latter then follows as a corollary by virtue of Lemma~\ref{lem:purification}. The statements refer to the trace distance, which is defined as $\Delta(\rho,\sigma):=\frac12\|\rho-\sigma\|_1$.

\begin{customthm}{1'} \label{thm:pure_definetti_state}
Let $\Psi_{S_1 \cdots S_{N} F}$ be a pure state on $\mathrm{Sym}^N(\h_S)\otimes\h_F$ for $N$ even. Let $\{\psi_{S|w}\}_{w \in \mathcal{W}}$ be a quantum $N$-design over $\h_{S}$. Define furthermore 
  \begin{align} \label{eq:conddF}
     \Psi_{S_1 \cdots S_{N/2} F|w}:=\frac{C_{N/2}}{|\mathcal{W}|p_w}\tr_{S_{N/2+1}\cdots S_N}\,(I_{S_1\cdots S_{N/2}F}\otimes\psi^{\otimes N/2}_{S|w}\cdot \Psi_{S_1 \cdots S_{N} F}) \ ,
  \end{align}
   where $p_w:=C_{N/2}\tr\,\psi^{\otimes N/2}_{S|w} \Psi_{S_1 \cdots S_{N/2}}/|\mathcal{W}|$ is a probability. Then, for any $n \leq N/2$, 
\begin{align} \label{eq:distanceexpectationvalue}
      \sum_{w\in\mathcal{W}} p_w \Delta(\Psi_{S_1 \cdots S_n F|w} \,,\, \Psi_{S_1|w}^{\otimes n} \otimes \Psi_{F|w} ) \le 3\sqrt{\frac{2\dim(\h_S)n}{N}} \ .
  \end{align}
\end{customthm}

\begin{proof}
Let us introduce an auxiliary state $\eta_{F|w}$ on $\h_F$ as an intermediate tool for the proof,
 \begin{align} \label{eq:condeta}
      \eta_{F|w}:=\frac{C_{N/2+n}}{|\mathcal{W}|\tilde{p}_w}\tr_{S_1\cdots S_N}\,(I_{S_1 \cdots S_{N/2-n} F}\otimes\psi^{\otimes N/2+n}_{S|w}\cdot \Psi_{S_1 \cdots S_{N}F}) \ ,
  \end{align}
 where $\tilde{p}_w:= C_{N/2+n} \tr\,\psi^{\otimes N/2+n}_{S|w} \Psi_{S_1 \cdots S_{N/2+n}}/|\mathcal{W}|$ is a probability.

Firstly, we verify that $\{p_w\}_{w \in \mathcal{W}}$ is a valid probability distribution. Using the fact that $\{\psi_{S|w}\}_{w\in\mathcal{W}}$ is a quantum $N$-design, that~\eqref{eq:povm}, for $k=N/2$, forms a POVM on the symmetric subspace $\mathrm{Sym}^{N/2}(\h_S)$, and that  $\Psi_{S_1 \cdots S_{N/2}}$ lies in that subspace, we find
\begin{equation}
    \sum_w p_w =\frac{C_{N/2}}{|\mathcal{W}|} \tr\bigl( \sum_w  \psi^{\otimes N/2}_{S|w} \Psi_{S_1 \cdots S_{N/2}}\bigr) = 1 \ .
\end{equation}
Similarly, $\{\tilde{p}_w\}_{w \in \mathcal{W}}$ is also  a valid probability distribution.

By applying the triangle inequality for the trace distance twice, we upper bound the distance on the left hand side of~\eqref{eq:distanceexpectationvalue} by a sum of three distances, 
\begin{multline}\label{eq:triangleineq}
     \sum_{w\in\mathcal{W}} p_w \Delta(\Psi_{S_1 \cdots S_nF|w},\,\Psi^{\otimes n}_{S_1|w}\otimes\Psi_{F|w}) \\
    \le \sum_{w\in\mathcal{W}} p_w \Delta(\Psi_{S_1 \cdots S_nF|w},\,\psi^{\otimes n}_{S|w}\otimes\eta_{F|w})
    +  \sum_{w\in\mathcal{W}} p_w \Delta(\Psi^{\otimes n}_{S_1|w}\otimes\eta_{F|w} ,\,\psi^{\otimes n}_{S|w}\otimes\eta_{F|w}) \\
    +\sum_{w\in\mathcal{W}} p_w \Delta(\Psi^{\otimes n}_{S_1|w}\otimes\eta_{F|w} ,\,\Psi^{\otimes n}_{S_1|w}\otimes\Psi_{F|w}) \ .
\end{multline}

We now show that each of the three terms above are bounded by $\sqrt{2\dim(\h_S)n/N}$. Expanding the first term gives
\begin{equation}\label{eq:fvdGinequality}
  \sum_{w\in\mathcal{W}} p_w \Delta(\Psi_{S_1 \cdots S_nF|w},\,\psi^{\otimes n}_{S|w}\otimes\eta_{F|w}) 
   \le\sqrt{1-\sum_{w\in\mathcal{W}} p_wF(\Psi_{S_1 \cdots S_nF|w},\,\psi^{\otimes n}_{S|w}\otimes\eta_{F|w})}
\end{equation}
where we have applied the Fuchs–van de Graaf inequality and Jensen's inequality. To bound the right hand side we compute the expectation value of the fidelity, 
\begin{equation}
    \begin{aligned}
&\sum_{w\in\mathcal{W}} p_w\left(\tr\sqrt{\psi^{\otimes n}_{S|w}\otimes\sqrt{\eta_{F|w}}\cdot\Psi_{S_1 \cdots S_nF|w}\cdot\psi^{\otimes n}_{S|w}\otimes\sqrt{\eta_{F|w}}}\right)^2 \\
=&\sum_{w\in\mathcal{W}} p_w\left(\tr\sqrt{\psi^{\otimes n}_{S|w}\otimes\left[\sqrt{\eta_{F|w}}\cdot\tr_{S_1\cdots S_n}\,(\psi^{\otimes n}_{S|w}\Psi_{S_1 \cdots S_nF|w})\cdot\sqrt{\eta_{F|w}}\right]}\right)^2 \\
=&\sum_{w\in\mathcal{W}} p_w\left(\tr\sqrt{\sqrt{\eta_{F|w}}\cdot\tr_{S_1\cdots S_n}\,(\psi^{\otimes n}_{S|w}\Psi_{S_1 \cdots S_nF|w})\cdot\sqrt{\eta_{F|w}}}\right)^2 \\
=&\frac{C_{N/2}}{|\mathcal{W}|}\sum_{w\in\mathcal{W}}\left(\tr\sqrt{\sqrt{\eta_{F|w}}\cdot\tr_{S_1\cdots S_N}\,(\psi^{\otimes N/2+n}_{S|w}\Psi_{S_1 \cdots S_NF})\cdot\sqrt{\eta_{F|w}}}\right)^2 \\
=&\frac{C_{N/2}}{C_{N/2+n}}\sum_{w\in\mathcal{W}}\,\tilde{p}_w\left(\tr\,\sqrt{\sqrt{\eta_{F|w}}\cdot\eta_{F|w}\cdot\sqrt{\eta_{F|w}}}\right)^2 \\
=&\frac{\binom{N/2+\dim(\h_S)-1}{N/2}}{\binom{N/2+n+\dim(\h_S)-1}{N/2+n}}\ge \left(\frac{N/2+1}{N/2+n+1}\right)^{\dim(\h_S)-1}\ge 1- \frac{2\dim(\h_S)n}{N} \ , 
\end{aligned}
\end{equation}
where the first equality follows as $\psi_{S|w}$ is a pure state (projector), and where we made the substitutions~\eqref{eq:conddF} and~\eqref{eq:condeta} to obtain the third and fourth equalities. The first term on the right hand side of~\eqref{eq:triangleineq} is thus bounded by $\sqrt{2\dim(\h_S)n/N}$. 

The second term in \eqref{eq:triangleineq} can be expanded as
\begin{equation}
\begin{aligned}
    &\sum_{w\in\mathcal{W}} p_w \Delta(\Psi^{\otimes n}_{S_1|w}\otimes\eta_{F|w} ,\,\psi^{\otimes n}_{S|w}\otimes\eta_{F|w}) 
    =\sum_{w\in\mathcal{W}} p_w \Delta(\Psi^{\otimes n}_{S_1|w} ,\,\psi^{\otimes n}_{S|w})\\
    \le&\sum_{w\in\mathcal{W}} p_w \sqrt{1-F(\Psi_{S_1|w} ,\,\psi_{S|w})^n}
    \le\sum_{w\in\mathcal{W}} p_w \sqrt{n(1- F(\Psi_{S_1|w} ,\,\psi_{S|w}))} \ ,
\end{aligned}
\end{equation}
where the last step follows from the Bernoulli's inequality.\footnote{Bernoulli's inequality asserts that for every integer $n\ge 0$ and any real number $x\ge -1$, $1+nx\le (1+x)^n$. Here we take $x=F(\Psi_{S_1|w} ,\,\psi_{S|w})-1$.}
A similar calculation as above gives
\begin{equation}
    \sum_{w\in\mathcal{W}} p_w \sqrt{n(1- F(\Psi_{S|w} ,\,\psi_{S|w}))} \le \sqrt{n}\sqrt{1-(1-2\dim(\h_S)/N)} = \sqrt{2\dim(\h_S)n/N}
\end{equation}
and it follows that the second term in \eqref{eq:triangleineq} is also bounded by $\sqrt{2\dim(\h_S)n/N}$. 

Finally, the third term in~\eqref{eq:triangleineq} can be bounded as
\begin{multline}
    \sum_{w\in\mathcal{W}} p_w \Delta(\Psi^{\otimes n}_{S_1|w}\otimes\eta_{F|w} ,\,\Psi^{\otimes n}_{S_1|w}\otimes\Psi_{F|w}) = \sum_{w\in\mathcal{W}} p_w \Delta(\eta_{F|w} ,\,\Psi_{F|w}) \\ 
    \le \sum_{w\in\mathcal{W}} p_w \Delta(\psi^{\otimes n}_{S|w}\otimes\eta_{F|w},\,\Psi_{S_1 \cdots S_nF|w})\le\sqrt{2\dim(\h_S)n/N} \ ,
\end{multline}
where the first inequality follows from the fact that the trace distance is non-increasing under the partial trace and the second inequality follows from the observation that  the expression is identical to the first term in~\eqref{eq:triangleineq}, which we have already bounded above. Then the claim \eqref{eq:distanceexpectationvalue} follows.
\end{proof}

For the proof of Theorem~\ref{thm:definetti_state}, we first note that, by assumption, there exists a permutation-invariant extension $\rho_{R_1 \cdots R_N E}$ of the given density operator.  Lemma~\ref{lem:purification} then implies the existence of a permutation-invariant purification $ \Psi_{S_1 \cdots S_N F}$ with $S_i\cong R_iR'_i$ and $F\cong EE'$. Given that a quantum $N$-design exists for any $N$ and $\dim(\h_S)$~\cite{seymour1984averaging,roy2009unitary}, we may now apply Theorem~\ref{thm:pure_definetti_state}. With $\rho_{R_1 \cdots R_{N/2}E|w}$ and $\rho_{E|w}$ defined as the corresponding marginals of the state~\eqref{eq:conddF}, and using that the trace distance is non-increasing under the partial trace, we find
\begin{align}
     \sum_{w\in\mathcal{W}} p_w\Delta(\rho_{R_1 \cdots R_n E|w}\otimes \ketbra{w}{w}_W \,,\, \rho_{R|w}^{\otimes n} \otimes \rho_{E|w}\otimes \ketbra{w}{w}_W) &\le 3\dim(\h_S)\sqrt{\frac{2n}{N}} \ ,
\end{align}
where we have appended a reference system $W$ with orthogonal projectors. By the strong convexity of the trace distance, the bound also implies

\begin{equation}\label{eq:secondlaststep}
     \Delta(\sum_{w\in\mathcal{W}} p_w\rho_{R_1 \cdots R_n E|w}\otimes \ketbra{w}{w}_W \,,\, \sum_{w\in\mathcal{W}} p_w\rho_{R|w}^{\otimes n} \otimes \rho_{E|w}\otimes \ketbra{w}{w}_W) \le 3\dim(\h_S)\sqrt{\frac{2n}{N}} \ .
\end{equation}
Finally, we define $\bar{\rho}_{R_1 \cdots R_{N/2} EW}:=\sum_{w\in\mathcal{W}} p_w\rho_{R_1 \cdots R_{N/2} E|w}\otimes \ketbra{w}{w}_W$. This state is an extension of $\rho_{R_1 \cdots R_{N/2} E}$. To verify this, we trace out~$W$. Using the fact that~\eqref{eq:povm} is a POVM on the symmetric subspace, we find
\begin{multline}
   \bar{\rho}_{R_1 \cdots R_{N/2} E} =\frac{C_{N/2}}{|\mathcal{W}|}\sum_{w\in\mathcal{W}} \tr_{R'_1\cdots R'_{N/2}S_{N/2+1}\cdots S_NE'}\,(\psi^{\otimes N/2}_{S|w}\cdot \Psi_{S_1 \cdots S_NF}) 
=\rho_{R_1 \cdots R_{N/2} E} \ .
\end{multline}
Theorem~\ref{thm:definetti_state} now follows immediately from~\eqref{eq:secondlaststep}.

\bibliographystyle{unsrt}
\bibliography{bh}
\end{document}